\documentclass{article}
\usepackage[margin=0.8in]{geometry}
\usepackage[utf8]{inputenc} 
\usepackage[T1]{fontenc}    
\usepackage{hyperref}       
\usepackage{url}            
\usepackage{booktabs}       
\usepackage{amsfonts}     
\usepackage{amsmath}
\usepackage{amsthm}  
\usepackage{nicefrac}       
\usepackage{microtype}      
\usepackage{lipsum}
\usepackage{mathtools}
\usepackage{threeparttable}
\usepackage{standalone}
\usepackage{booktabs, dcolumn}
\usepackage{pdflscape}
\usepackage{tablefootnote}
\usepackage{booktabs,caption,fixltx2e}
\usepackage{authblk}
\usepackage{natbib}
\usepackage{threeparttable}

\newcommand{\bd}[1]{#1}
\newcommand{\argmin}{\text{argmin}}
\newcommand{\diag}{\text{diag}}
\newtheorem{theorem}{Theorem}
\newcommand{\beginsupplement}{%
        \setcounter{table}{0}
        \renewcommand{\thetable}{S\arabic{table}}%
        \setcounter{figure}{0}
        \renewcommand{\thefigure}{S\arabic{figure}}%
     }

\theoremstyle{definition}
\newtheorem{example}{Example}[section]

\title{Generalized Meta-Analysis for Multiple Regression Models
Across Studies with Disparate Covariate Information}
\author[1]{Prosenjit Kundu}
\author[1]{Runlong Tang}
\author[1,2]{Nilanjan Chatterjee}
\affil[1]{Department of Biostatistics, Bloomberg School of Public Health, The Johns Hopkins University}
\affil[2]{Department of Oncology, School of Medicine, The Johns Hopkins University}

\begin{document}
\maketitle

\begin{abstract}
Meta-analysis, because of both logistical convenience and statistical efficiency,
is widely popular for synthesizing information on common parameters of interest across multiple studies.
We propose developing a generalized meta-analysis 
approach for combining information on multivariate regression parameters
across multiple different studies which have varying level of covariate information.
Using algebraic relationships between regression parameters in different dimensions,
we specify a set of moment equations for estimating parameters of a maximal model through information available
from sets of parameter estimates from a series of reduced models available from the different studies.
The specification of the equations requires a reference dataset to estimate the joint distribution of the covariates.
We propose to solve these equations using the generalized method of moments approach,
with the optimal weighting of the equations taking into account uncertainty associated with estimates of the parameters of the reduced models.
We describe extensions of the iterated reweighted least square algorithm
for fitting generalized linear regression models using the proposed framework. Based on the same moment equations, we also propose a diagnostic test for detecting violation of underlying model assumptions, such as those arising due to heterogeneity in the underlying study populations.
Methods are illustrated using extensive simulation studies and a real data example involving
the development of a breast cancer risk prediction model using disparate risk factor information from multiple studies.
\end{abstract}

\section{Introduction}
In a variety of domains of applications,
including observational epidemiologic studies,
clinical trials and modern genome-wide association studies,
meta-analysis is widely used to synthesize information on underlying common parameters of interest across multiple studies \citep{laird86, laird15, ioannidis, kavvoura}.
The popularity of meta-analysis stems from the fact
that it can be performed based only on estimates of model parameters and standard errors,
avoiding various logistical,
ethical and privacy concerns associated with accessing the individual level data that is required in pooled analysis.
Moreover, in many common settings, it can be shown that under reasonable assumptions, meta-analyzed estimates of model parameters are asymptotically as efficient as those from pooled analysis
\citep{olkin, mathew, zeng}. In fact, meta-analysis approaches are now being used in divide and conquer approaches to big data, even when individual level data are potentially available, because of the daunting computational task of model fitting with extremely large sample sizes \citep{michael, liu, chun}.

In this article, we study the problem of multivariate meta-analysis in the setting of parametric regression modeling of an outcome given a set of covariates.
In standard settings, if estimates of multivariate parameters for an underlying common regression model and associated covariances are available across all the studies, then meta-analysis can be performed by taking inverse-variance-covariance weighted average of the vector of regression coefficients \citep{van, ritz, jackson}.
In many applications, a typical problem is that different studies include different, but possibly overlapping, sets of covariates.
In a large consortium of epidemiologic studies,
for example, some key risk factors will be measured across all the studies.
Inevitably, however, there will be potentially important covariates
which are measured only in some, but not all the studies.
It is also possible that some covariates are measured at a more detailed level or with a finer instrument in some studies compared to others.
Disparate sets of covariates across studies render standard meta-analysis to be applicable
for the development of models only limited to a core set of variables that are measured in the same fashion across all the studies.

We propose a generalized meta-analysis (GENMETA) approach for building rich models using information on model parameters across studies with disparate covariate information.
GENMETA is built upon a fundamental mathematical relationship between parameters of two regression models in different dimensions
from our recent study \citep{chatterjee}.
In the current article, we utilize this mathematical relationship to develop a general framework for combining information on parameters of various models of different dimensions within the generalized method of moments framework \citep{hansen, imbens}.
We develop an iterated reweighted least square  algorithm allowing stable and speedy computation of estimates.
The proposed method requires access to a reference dataset for estimating of joint distribution of the covariates in a nonparametric fashion. We show how the reference dataset can be used to derive an optimal estimator and associated variance-covariances even when entire variance-covariance matrices for model parameter estimates may not be obtainable from individual studies.

\section{Models and Methods}
\subsection{Model Formulation}
\label{sec:Formulation}
Suppose we have parameter estimates $\hat{\theta}_{k}$
and associated estimates of their covariance matrices $S_{k}$
from $K$ independent studies
which have fitted  reduced regression models, for which the likelihood is of the form $g_{k}(Y|X_{A_{k}}; \theta_{k})$,
where $Y$ is a common underlying outcome of interest
but the vector of  covariates $X_{A_k}$ is potentially distinct across the studies.
Let $X$ be the set of covariates used across all studies
and we assume the true distribution of $Y$ given $X$ can be specified
by a maximal regression model $f(Y|X;\beta)$.
Our goal is to estimate and make inference about $\beta^{*}$,
the true value of $\beta$,
based on summary-level information, $(\hat{\theta}_{k}, S_{k})$ from the $K$ studies.

In the proposed setup, it is possible but not necessary,
that one or more of the studies have information on all covariates
to fit the maximal model by themselves.
Under certain study designs, such as the multi-phase designs \citep{breslow88, breslow97, whittemore, wild} and the partial questionnaire design \citep{carroll}, data could be partitioned into independent sets where the maximal model can be fitted on some sets and various reduced models can be fitted on others.
The maximal model $f(Y|X;\beta)$ and the reduced models $g_{k}(Y|X_{A_k};\theta_k)$ may have different parametric forms, such as logistic and probit models when $Y$ is a binary disease outcome.
This setup also allows incorporation of covariates which may be measured more accurately or in a more refined fashion
in some studies than others.
For example, different studies may include two types of measurements, namely, $Z_1$ and $Z_2$, for the same covariate, with $Z_2$ a more refined measurement. In this case the different reduced models may include $Z_1$ or $Z_2$, but we require that the reference dataset includes both $Z_1$ and $Z_2$. In the maximal model, we can enforce that Y is independent of $Z_1$ given $Z_2$ by setting the regression parameters associated with $Z_1$ to be zero.
If all of reduced models were the same,
i.e. all studies have the same covariate information,
we have $X_k =X$, $\theta_{k}=\beta$ and $g_{k}=f$ for each $k$,
and the common parameter of interest $\beta^{*}$ can be efficiently estimated
by the fixed-effect meta-analysis estimator
$\hat{\beta}_{\mbox{meta}} = \sum_{k=1}^K(\sum_{k=1}^K S_{k}^{-1})^{-1}S_{k}^{-1}\hat{\theta}_{k}$,
the variance of which, in turn, can be estimated by
$\hat{\Sigma}_{\mbox{meta}} = (\sum_{k=1}^K S_{k}^{-1})^{-1}$
\citep{van, ritz, jackson}.

\subsection{A Special Case Involving Linear Regression Model}
As readers may find it counter-intuitive to comprehend how it is possible to estimate parameters of the maximal model as no single study may have ascertained $Y$ and all components of $X$ simultaneously, following we give a linear model example to help develop insight into the problem.
Suppose, one is interested in developing a multiple linear regression model for $Y$ based on a set of covariates $X$ in the form
\[ Y= \alpha + \sum_{k=1}^{K}\beta_{k}X_k + \epsilon\]
where it is further assumed that $\epsilon \sim N(0,\sigma^2)$. Without loss of generality, we will assume that all the variables $Y$, $X_1,\ldots,X_K$ are standardized to have mean zero and variance one. Under this model, the population parameter $\beta=(\beta_1,\ldots,\beta_K)^{T}$ can be expressed as $\beta= {E(X^{T}X)}^{-1}E(X^{T}Y)=R^{-1}E(X^{T}Y)$, where $R$ is the population correlation matrix of $X$. Now, suppose we have no data available on $Y$ and mutivariate $X$ on the same sample, but we have estimates available for parameters ($\theta_k$, $k=1,...K$) for univariate linear regression models of the form
\[ Y = \theta_{k}X_{k} +\psi_k. \]

From above $\theta_{k}=E(X_kY)$ and thus $\hat\theta=(\hat\theta_1,\ldots,\hat\theta_K)$, provides an estimate of  the cross product terms, $E(X^{T}Y)$, which is required in estimating $\beta$. Further, if we have a reference dataset, which has information on multivariate $X$ but is not required to be linked to $Y$,  it can be use to  estimate $R$, as $\hat{R}$ say, and a consistent estimate of $\beta$ can be obtained simply as $\hat\beta = \hat{R}^{-1}\hat\theta$. Thus, this simple derivation shows that it is possible to estimate parameters of a multiple regression model using information on parameters of a series of univariate regression models and a reference dataset. In fact, this observation that information on univariate regression parameters (known as summary-level statistics) can be utilized to reconstruct estimates of parameters of multivariate regression model has revolutionized the field of statistical genetics. Recently, a large variety of methods have been developed for the inference on parameters underlying multivariate regression models utilizing widely available summary-level results from large GWAS and reference datasets to estimate linkage disequlibrium across genetic markers \citep{yangvisscher2012, bulik-sullivan_ld_2015, zhu_integration_2016, pasaniuc}. In the following, we show a more general statistical formulation of the problem that allows consideration of non-linear models and use of information from arbitrary types of reduced models as opposed to simply univariate models.

\subsection{Generalized Meta-Analysis}
\label{sec:Meta-GMM Estimation}
The key idea underlying the proposed generalized meta-analysis is
that we convert information on parameters from reduced models
into a set of equations that are informative about the parameters of the maximal model.
In the following, we will make the assumptions (i)  the same probability law for $(Y, X)$ holds for all the underlying populations (ii) $f(Y|X;\beta)$ is a correctly specified model for the conditional distribution of $(Y|X)$ and (iii) we have a reference dataset to estimate empirically the joint distribution of all the factors included in $X$. 

Here, we assume all the studies employ a random sampling design
and the same probability law for $(Y, X)$ holds for all the underlying populations.
Let $s_{k}(y|x_{A_k}; \theta_{k}) = \partial
\log g_{k}(y|x_{A_k}; \theta_{k})/\partial \theta_{k}$ be the score function of the $k$th reduced model and denote $u_{k}(x;\beta, \theta_{k}) = \int s_{k}(y|x_{A_k}; \theta_{k})f(y|x;\beta)dy$. Assume $\hat{\theta}_{k}$ is the maximum likelihood estimator from the $k$th study and denote $\theta_{k}^{*}$ as the asymptotic limit of $\hat{\theta}_k$.
Irrespective of whether the reduced models are correct or not, $E_{pr^*}\{s_{k}(Y|X_{A_k}; \theta_{k}^{*})\}=0$ holds,
where $pr^*$ denotes the true probability law.
Assuming the maximal model is correctly specified,
we can write $pr^*(Y,X_{A_k})=\int_{X\setminus X_{A_k}}f(Y|X;\beta^*)dF^{*}(X)$.
Hence, a general equation establishing relationship between $\beta^{*}$ and $\theta_{k}^{*}$ is in the form \citep{chatterjee}
\begin{equation}\label{equ:generic moment function}
    \int u_{k}(x;\beta^{*},\theta_{k}^{*})d F^{*}(x)=0.
\end{equation}

As we may not have individual level data from the studies,
the above equations cannot be directly evaluated.
Instead, we assume that we have a reference sample of size $n$,
independent of the study samples,
on which measurements on $X$ are available.
The reference sample need not to be linked
with the outcome $Y$ of interest
and its sample size can be fairly modest compared to the study sample sizes.

With $\hat{\theta}_{k}$ from the studies
and the reference sample $\{\bd{X}_{i}\}_{i=1}^{n}$,
we can set up the estimating equations
$ U_{n}(\beta,\hat{\theta})=(1/n)\sum_{i=1}^{n}U(X_i; \beta,\hat{\theta})=0$,
where
$U(x;\beta,\theta)=
(u_{1}^T(x;\bd{\beta}, \bd{\theta}_{1}),
\ldots,
u_{K}^T(x;\bd{\beta}, \bd{\theta}_{K}))^{T}$,
$\hat{\bd{\theta}}=(\hat{\bd{\theta}}_{1}^{T}, \ldots, \hat{\bd{\theta}}_{K}^{T})^{T}$
and ${\bd{\theta}}=({\bd{\theta}}_{1}^{T}, \ldots, {\bd{\theta}}_{K}^{T})^{T}$.
Denote the dimensions of $\bd{\theta}_{k}$ and $\bd{\beta}$ as $d_{k}$ and $p$, respectively.
Because the number of equations $d=\sum_{k=1}^{K}d_{k}$ can be larger than the number of unknown parameters $p$,
the estimating equations may not be solved exactly.
Following the technique of generalized method of moments ,
we propose the following generalized meta-analysis (GENMETA) estimator of $\bd{\beta}^{\star}$:
$
\hat{\bd{\beta}}
=
{\argmin}_{\bd{\beta}}Q_{\hat{C}}(\bd{\beta})$
where
$
Q_{\hat{C}}(\bd{\beta})
=
U_{n}(\bd{\beta}, \hat{\bd{\theta}})^{T}\hat{C}U_{n}(\bd{\beta}, \hat{\bd{\theta}})
$
and $\hat{C}$ is a positive semi-definite weighting matrix.
Following the well established theory of generalized method of moments \citep{hansen,engle},
we derive the asymptotic properties of the GENMETA estimator.
Assume the study summary statistics $\hat{\bd{\theta}}_{k}$ are independent across studies;
${n_{k}}^{1/2}(\hat{\bd{\theta}}_{k}- \bd{\theta}^{*}_{k})
{\rightarrow}N(0, \bd{\Sigma}_{k})$ in distribution
and
${\lim}_{n\rightarrow\infty} n_{k}/n = c_{k}>0$ for each $k$;
and the reference sample is independent of the study samples.
Denote
$\bd{\Gamma}=
E\{{\partial}
U(X;\bd{\beta}, \bd{\theta}^{\star})/{\partial \bd{\beta}}\mid_{\bd{\beta}=\bd{\beta}^{*}}\}$,
$\bd{\Delta}
=
E\{U(X;\bd{\beta}^{*},\bd{\theta}^{\star})
U^{^T}(X;\bd{\beta}^{*},\bd{\theta}^{\star})\}$
and
$\bd{\Lambda} =
\diag(\bd{\Lambda}_1, \ldots, \bd{\Lambda}_K)$,
where
$\bd{\Lambda}_{k}=(1/c_{k})W_{k}\bd{\Sigma}_{k}W_{k}^{T}$
and
$W_{k} = E\{{\partial}
u_{k}(X; \bd{\beta}^{*}, \bd{\theta}_{k})/{\partial \bd{\theta}_{k}}\}
|_{\bd{\theta}_{k} = \bd{\theta}_{k}^{*}}$
for $k=1,\dots,K$.

\begin{theorem}[Consistency and Asymptotic Normality of $\hat{\bd{\beta}}$]
\label{theorem:asymptotic normality of beta hat}
Suppose the positive semi-definite weighting matrix
$\hat{C} {\rightarrow} C$ in probability. Then, under Assumptions (A1)-(A4) in the appendix, $\hat{\bd{\beta}} {\rightarrow} \bd{\beta}^{*}$ in probability.
Further, given $\beta^{*}$ is an interior point and under additional Assumptions (A5)-(A9) in the appendix,
${n}^{1/2}(\hat{\bd{\beta}}-\bd{\beta}^{*})$
converges in distribution to the normal distribution
$N(0, (\bd{\Gamma}^{T}C\bd{\Gamma})^{-1}\bd{\Gamma}^{T}C
(\bd{\Delta}+\bd{\Lambda})C\bd{\Gamma}(\bd{\Gamma}^{T}C\bd{\Gamma})^{-1})
$.
\end{theorem}

The optimal $C$ that minimizes the above asymptotic covariance matrix is $C_{\text{opt}}=(\bd{\Delta}+\bd{\Lambda})^{-1}$
and the corresponding optimal asymptotic covariance matrix is $\{\bd{\Gamma}^{T}(\bd{\Delta}+\bd{\Lambda})^{-1}\bd{\Gamma}\}^{-1} $.
Because $C_{\text{opt}}$ itself depends on unknown
underlying parameters, it requires iterative evaluation.
In our applications,
we first evaluate an initial GENMETA estimator
with a simple choice of $\hat{C}$ such as the identity matrix.
We then obtain the iterated GENMETA estimator
by continuing to set $\hat{C}=\hat{C}_{\text{opt}}$ based on the latest parameter estimate till convergence.
By Theorem \ref{theorem:asymptotic normality of beta hat},
$\hat{\bd{\beta}}$ with $C_{\text{opt}}$ approximately follows
a Gaussian distribution with mean $\bd{\beta}^{*}$
and covariance matrix
\begin{equation}\label{equ:covariance matrix}
    [\bd{\Gamma}^{T}\{\frac{1}{n}\bd{\Delta}
+ \diag(
\frac{1}{n_{1}}W_{1}\bd{\Sigma}_{1}W_{1}^{T},
\ldots,
\frac{1}{n_{K}}W_{K}\bd{\Sigma}_{K}W_{K}^{T})^{-1}\}\bd{\Gamma}]^{-1},
\end{equation}
which indicates that the precision of GENMETA depends on the size of the reference sample, $n$,
as well as on those of the studies, $n_{k}$.
However, as we will see in Section 3, the study sample sizes are the  dominating  factor controlling the precision of GENMETA
and with fixed $n_{k}$'s, the precision of GENMETA quickly reaches plateau as a function of $n$.

For the implementation of the optimal GENMETA and the variance estimation of any of the GENMETA estimators,
one needs to have valid estimates of $\Lambda_{k}$,
which depend on $\Sigma_k$,
the asymptotic covariance matrices of the estimates of the reduced model parameters.
Ideally,  the studies should provide robust estimates of the covariance matrices,
such as the sandwich covariance estimators,
so that they are valid irrespective of
whether the underlying reduced models are correctly specified or not.
In practice, however, while some kind of estimates of standard errors of the individual parameters are expected to be available from a study,
obtaining the desired robust estimate of the entire covariance matrix could be difficult.
When no estimate of $\bd{\Sigma}_k$ is available from the $k$th study,
one can take the advantage of the reference sample to estimate it by
$\hat{\bd{\Sigma}}_{k}^{\text{ref}}=\hat{J}^{-1}\hat{V}\hat{J}^{-1}$,
where
$\hat{J}=P_n [E_{Y\mid\bd{X}}\{\nabla_{\bd{\theta}_{k}}s_{k}(\theta_{k})\}]|_{\theta_k = \hat{\theta}_k}$,
$\hat{V}=P_n [E_{Y\mid\bd{X}} \{s_{k}(\theta_{k})s_{k}(\theta_{k})^{T}\}]|_{\theta_k = \hat{\theta}_k}$,
$s_{k}(\hat{\theta}_{k})=s_{k}(Y\mid X_{A_k}; \theta_{k})|_{\theta_k = \hat{\theta}_k}$,
$\hat{\theta}_{k}$ is a consistent estimator of $\bd{\theta}_{k}^{\star}$,
$\hat{E}_{Y\mid\bd{X}}$ is the expectation with respect to the distribution of $Y\mid\bd{X}$
with $\bd{\beta}^{\star}$ replaced by a consistent estimator $\hat{\beta}$,
and $P_n$ is the empirical measure with respect to the reference sample.
Further, assuming
$E_{Y\mid\bd{X}}\{\nabla_{\bd{\theta}_{k}}s_{k}(\bd{\theta}_{k})\}|_{\theta_k = \theta^*_k}
=
\nabla_{\bd{\theta}_{k}}E_{Y\mid\bd{X}}\{s_{k}(\bd{\theta}_{k})\}|_{\theta_k = \theta^*_k}$,
it follows
$\bd{\Lambda}_{k}=(1/c_{k})E_{(Y,\bd{X})} \{s_{k}(\bd{\theta}_{k})s_{k}(\bd{\theta}_{k})^{T}\}|_{\theta_k = \theta^*_k}$,
which can be estimated by
$\hat{\bd{\Lambda}}_{k}^{\text{ref}}=
(1/c_{k})P_n [E_{Y\mid\bd{X}} \{s_{k}(\theta_{k})s_{k}(\theta_{k})^{T}\}]|_{\theta_k = \hat{\theta}_k}$.
For example,
suppose $Y\mid\bd{X}$ and $Y\mid\bd{X}_{A_{k}}$ follow logistic distributions
with parameters $\bd{\beta}^{\star}$ and $\bd{\theta}_{k}$, respectively.
Denote $\bd{X}=(1, \bd{X}^{T})^{T}$ and $\bd{X}_{A_{k}}=(1, \bd{X}_{A_k}^{T})^{T}$.
Then,
\begin{equation}\label{equ:Lambda.k.hat.ref}
\hat{\bd{\Lambda}}_{k}^{\text{ref}}
=
\frac{1}{c_{k}}
P_n
[
\{
(1+e^{\bd{X}_{A_{k}}^{T}\hat{\bd{\theta}}_{k}})^{-2}
(1+e^{-\bd{X}^{T}\bd{\beta}})^{-1}
+
(1+e^{-\bd{X}_{A_{k}}^{T}\hat{\bd{\theta}}_{k}})^{-2}
(1+e^{\bd{X}^{T}\hat{\bd{\beta}}})^{-1}\}
\bd{X}_{A_{k}}\bd{X}_{A_{k}}^{T}]
.
\end{equation}
In section 3, we will study the properties of the GENMETA estimators
using either covariance matrices estimated from studies or the reference sample.

It is insightful to explore the connection between GENMETA and standard meta-analysis
when all of the reduced models are identical to the maximal model,
that is, when
$\bd{\theta}^*_{k} = \bd{\beta}^*$,
$X_{A_k} = X$
and $g_{k}=f$ for each $k$.
Under this setup, the moment vector evaluated at the true parameters becomes zero for each study,
i.e. $u_{k}(\bd{X}; \bd{\beta}^*,\bd{\theta}^*_{k})=u_{k}(\bd{X}; \bd{\beta}^*,\bd{\beta}^*)
= 0$.
This simplification implies $\bd{\Delta} = \vec{0}$
and thus the optimal weighting matrix is
$C_{\text{opt}} = \bd{\Lambda}^{-1}
= \diag(c_1\bd{\Sigma},\dots,c_K\bd{\Sigma})$,
where $\bd{\Sigma}$ is the inverse of the Fisher's information matrix of $f$.
Denote by $\hat{\bd{\beta}}_{\text{opt}}$ the GENMETA estimator with 
a consistent estimator of $C_{\text{opt}}$.
Then, by arguments similar to those in the proof of Theorem \ref{theorem:asymptotic normality of beta hat},
$\hat{\bd{\beta}}_{\text{opt}}$ can be expressed as
\begin{equation*}\label{eqn:asyequv}
\hat{\bd{\beta}}_{\text{opt}}
=\hat{\bd{\beta}}_{\text{meta}}+ o_p(1/{n}^{1/2}),
\end{equation*}
which implies that $\hat{\bd{\beta}}_{\text{opt}}$ and $\hat{\bd{\beta}}_{\text{meta}}$
are asymptotically equivalent in terms of limiting distributions.

\subsection{Generalized Linear Model and Iterated Reweighted Least Square Algorithm}
GENMETA computation involves minimization of a quadratic form,
$Q_{{C}}(\bd{\beta})=U_n^T(\bd{\beta},\hat{\bd{\theta}})CU_n(\bd{\beta},\hat{\bd{\theta}})$,
with a known weighting matrix $C$.
Next, we derive the iterated reweighted least squares algorithm
for minimizing the quadratic form,
assuming that the maximal and reduced models belong to the class of generalized linear models [\citep{mccullagh}].
Specifically,
the densities of $Y\mid\bd{X}$ and $Y\mid\bd{X}_{A_{k}}$ are of the forms
$\exp(\{1/a(\phi)\}(yh(\bd{x}^{T}\bd{\beta}^{\star}) - b \{h(\bd{x}^{T}\bd{\beta}^{\star})\}) + c(y;\phi))$
and
$\exp(\{1/a(\phi_{k})\}(yh(\bd{x}_{A_{k}}^{T}\bd{\theta}_{k}) - b \{h(\bd{x}_{A_{k}}^{T}\bd{\theta}_{k}\}) + c(y;\phi_{k}))$,
respectively,
where $a(\cdot)$, $b(\cdot)$ and $c(\cdot)$ are known functions,
$h(\cdot)=b'^{-1}\{g^{-1}(\cdot)\}$,
$g$ is a monotone and differentiable link function,
and $\phi$ and $\phi_{k}$ are the dispersion parameters
of the maximal and the $k$th reduced models, respectively.

First, we assume that the dispersion parameters, $\phi$ and $\phi_{k}$'s, are known
and later we will relax this assumption.
For this case,
it follows, for each $k$,
\begin{equation}\label{eqn: IRLS_score}
u_{k}(x;\bd{\beta}, \bd{\theta}_{k})
= r_{k}(x;\bd{\beta},\bd{\theta}_{k}, \phi_{k})x_{A_k},
\end{equation}
where
$r_{k}(x;\bd{\beta},\bd{\theta}_{k}, \phi_{k})
=
\{1/a(\phi_k)\}(g^{-1}(x^T\bd \beta) - g^{-1}(x^T_{A_k}\bd \theta_{k}))
h'(x^T_{A_k}\bd \theta_{k})$.
Then,
the empirical moment vector is
$U_n(\bd\beta, \hat{\bd{\theta}}) =
P_n
(
u_{1}(X;\bd{\beta}, \hat{\bd{\theta}}_{1})^{T}, \ldots,
\vec{u}_{K}(X;\bd{\beta}, \hat{\bd{\theta}}_{K})^{T}
)^T$.
The Newton-Raphson (NR) method for searching the minimizer of $Q_{{C}}(\bd{\beta})$ can be written as
\begin{equation} \label{eqn:IRLS iter}
\bd \beta^{(t+1)}  = \bd \beta^{(t)} - (X_{\text{rbind}}^TW^*X_{\text{rbind}})^{-1}X^T_{\text{rbind}}WX_{A_{\text{diag}}}CX^T_{A_{\text{diag}}}r
\end{equation}
where $X_{\text{rbind}} = \bd{1} \otimes X$ and $X_{(n\times p)}$ is the reference data matrix;
$X_{A_{\text{diag}}} = \text{diag}(X_{A_1}, \hdots, X_{A_K})$
and $X_{A_k (n\times d_{k})}$ is the reference data matrix for the $k$th study;
$W = \diag(\bd{W}_{1}, \ldots, \bd{W}_{K})$,
$\bd{W}_{k} = \text{diag}(w_{k1},\dots,w_{kn})$
and $w_{ki}
=
(1/(a(\phi_k)g^\prime \{ g^{-1}(X^T_i\bd{\beta}^{(t)})))h'(X^T_{A_k,i}\hat{\bd{\theta}}_{k})$ for $k = 1,\dots,K$; $i= 1,\dots, n$;
$W^*$ is the sum of $W X_{A_{diag}}C X^T_{A_{diag}}W$ and $ \text{diag}(r^TX_{A_{diag}}CX^T_{A_{diag}}L)$,
a diagonalized matrix from a vector;
$\bd{r}=(\bd{r}_{1}^{T},\ldots, \bd{r}_{K}^{T})^{T}$,
$\bd{r}_{k}=(r_{k1}, \ldots, r_{kn})^{T}$
and $r_{ki}=r_{k}(\vec{X}_{i};\bd{\beta}^{(t)},\hat{\bd{\theta}}_{k}, \phi_{k})$;
and
$\vec{L} = \diag(\bd{L}_{1},\ldots, \bd{L}_{K})$,
$\bd{L}_{k} = \diag(l_{k1},\ldots,l_{kn})$
and
$l_{ki}
=
-g'' \{ g^{-1}(\bd{X}_{i}^{T}\bd{\beta}^{(t)})\}/(a(\phi_{k})[g' \{ g^{-1}(\bd{X}_{i}^{T}\bd{\beta}^{(t)})\}]^{3}h'(X^T_{A_k,i}\hat{\bd{\theta}}_{k}))$.
Equation (\ref{eqn:IRLS iter}) implies that the Newton-Raphson's method
is an iterated reweighted least squares algorithm.

When $\phi$ and $\phi_{k}$'s are unknown,
we propose to first obtain the GENMETA estimator $\hat{\bd{\beta}}$ of $\bd{\beta}^{\star}$ as above
with $\phi_{k}'s$ replaced by $\hat{\phi}_{k}$'s.
Next, we consider the estimation of $\phi^{\star}$, the true value of $\phi$.
For the $k$th reduced model,
we have an additional score function with respect to $\phi_{k}$,
from which,
similar to equation (\ref{eqn: IRLS_score}),
we can obtain
\begin{equation*}
u_{k}(X;\bd{\beta}, \phi, \bd{\theta}_{k}, \phi_{k})=
-\frac{a^\prime(\phi_k)}{a^2(\phi_k)}
(g^{-1}(X^T\bd{\beta})h(\bd{X}_{A_{k}}^{T}\bd{\theta}_{k}) - b \{ h(\bd{X}_{A_{k}}^{T}\bd{\theta}_{k})\})
+ q_{k}(X; \bd{\beta}, \phi, \phi_k),
\end{equation*}
where $q_{k} = E_{Y\mid\bd{X}}\{c'(Y;\phi_k)\}$ and $c'(Y;\phi_k)$ is the derivative of $c(Y;\phi_k)$ with respect to $\phi_{k}$.
Then,
the empirical moment vector for $\phi$ is
$U_n(\phi) =
P_n
(
u_{1}(X; \hat{\bd{\beta}}, \phi, \hat{\bd{\theta}}_{1}, \hat{\phi}_{1})^{T},
\ldots,
u_{K}(X; \hat{\bd{\beta}}, \phi, \hat{\bd{\theta}}_{K}, \hat{\phi}_{K})^{T}
)^{T}$.
To estimate $\phi^{\star}$,
we need to compute the minimizer of $U_n(\phi)^{T}\bd{C} U_n(\phi)$,
where $C$ is a known weighting matrix.
The Newton-Raphson steps can be written as
\begin{equation} \label{eqn:IRLS iter dip}
\phi^{(t+1)}
= \phi^{(t)} - J_n^{-1}(\phi^{(t)})
D_{n}(\phi^{(t)}),
\end{equation}
where
$J_n(\phi) = U^T_n(\phi)C{d^2 } q_n(\phi)/{d\phi^2}
+ ({d}q_n(\phi)/{d\phi})^TC{d}q_n(\phi)/{d\phi}$,
$D_{n}(\phi)=U^T_n(\phi^{(t)})C{d}q_n(\phi)/{d\phi}$
and $q_n(\phi) =
P_n
(
q_{1}(X; \hat{\bd{\beta}}, \phi, \hat{\phi}_1), \ldots, q_{K}(X; \hat{\bd{\beta}}, \phi, \hat{\phi}_K)
)^T$.
In brief, when $\phi$ and $\phi_{k}$, $k=1,\dots,K$, are unknown,
we first choose initial estimates $\bd \beta^{(0)}$ and $\phi^{(0)}$.
Then, we get the GENMETA estimator $\hat{\bd \beta}$ by using equation (\ref{eqn:IRLS iter}) until a stopping rule is reached.
Subsequently, $\phi^{(0)}$, $\hat{\bd \beta}$ and the study estimates are plugged in equation (\ref{eqn:IRLS iter dip})
and the process is repeated until a stopping rule is reached to get the GENMETA estimator of $\phi^{*}$.
In each NR step, the weighting matrix $\bd{C}$ is estimated by the estimates from the previous step.

\subsection{Diagnostic Test for Model Violation}
GENMETA relies on several modeling assumptions, including homogeneity of the underlying populations with respect to the distribution of covariates and regression parameters, and correct specification of the maximal model. In the absence of individual level data from the different studies, these assumptions could not be tested in the usual manner using traditional diagnostic tests. However, even with summary-level data, some diagnostic testing is possible. In particular, from an intuitive perspective, departure of the GENMETA estimating equations, when evaluated at estimated parameter values, from their expected null value will be indicative of disagreement between the model and the observed data, i.e. the estimates of the parameters from the reduced models from different studies. For example, if the regression parameters underlying the maximal model are highly heterogeneous across studies, then the assumption of a common $\beta$ in GENMETA will not be able to explain the heterogeneity that is expected to be present in overlapping reduced model parameters across the studies. 
Specifically, we propose to use the score test based on the statistic, $T_{GENMETA}=nQ_{\hat{C}_{opt}}(\hat{\beta})$, where $\hat{\beta}$ is the GENMETA estimate. When all the underlying assumptions are correct, from the standard generalized method of moments theory, $T_{GENMETA}$ converges in distribution to a $\chi^2$ distribution with $d-p$ degrees of freedom, where $d$ is the total number of GENMETA equations and $p$ is the total number of underlying parameters that are being estimated. The test is only applicable when $d > p$, which arises when different studies have overlapping covariates. 

\section{Simulations}
We study the performance of the GENMETA estimators through simulation studies in both idealized and non-idealized settings. In all simulations, we assume that the relationship between a binary outcome variable $Y$ and three covariates $(X_1,X_2,X_3)$
can be described by a logistic regression model of the form
\begin{equation}\label{Equation:Simu}
    Y\mid(X_{1}, X_{2}, X_{3}) \sim \text{Bernoulli}([1+\exp\{-(\beta_{0}^{*}+\beta_{1}^{*}X_{1}+
\beta_{2}^{*}X_{2}+\beta_{3}^{*}X_{3})\}]^{-1})
\end{equation}
where $(X_1,X_2,X_3)$ follows a multivariate normal distribution with mean $\mu = (\mu_1, \mu_2, \mu_3)$, variance $\sigma^2 = (\sigma_1^2, \sigma_2^2, \sigma_3^2)$ and underlying correlations
$\rho=(\rho_{12},\rho_{13},\rho_{23})$.
We chose
$\beta_1^*=\beta_2^*=\beta_3^*=\log 1.3$ to reflect a moderate degree of association of the outcome with each covariate after adjusting for the others.
We assume existence of three separate studies,
where each study fits a reduced logistic model for the outcome $Y$ on two of the covariates in the form
\begin{equation}\label{Equation:Simureduced}
 Y\mid(X_{i}, X_{j}) \sim \text{Bernoulli}((1+\exp\{-(\theta_{0,ij}^{*}+\theta_{i,ij}^{*}X_{i}+\theta_{j,ij}^{*}X_{j},
)\})^{-1}),
\end{equation}
with  $X_1$ and $X_2$ included in Study-I, $X_2$ and $X_3$ in Study-II and $X_1$ and $X_3$ in Study-III. Here, as data for each study are generated using the maximal model, the reduced models are by definition incompatible due to non-collapsibility of the logistic model. We fix the sample size of the studies at $n_1=300$, $n_2=500$ and $n_3=1000$ and vary the sample size of the reference dataset.

\subsection{Homogeneous Population}
We assume that the studies are conducted in the same underlying population from which the reference sample is drawn.
Under this setting, there exists a common mean vector $\mu_b = (0,0,0)$, common variance vector $\sigma_b^2 = (1,1,1)$ and common correlation vector $\rho_{b}=0.3,0.6,0.1)$,
that describes the joint distribution of the three covariates across all the underlying populations.
In the first set of simulation,
we assume a fixed sample size $n=50$ for the reference dataset. 
In all settings, we simulate data $(Y,X_1,X_2,X_3)$ for the underlying studies based on the data generating models as described above and
fit the respective reduced models to obtain estimates of the reduced model parameters.
For each set of simulated data,
we obtain estimates of covariance matrices of the reduced model parameters
using robust sandwich estimators based on either the study datasets themselves,
or the reference dataset (see (\ref{equ:Lambda.k.hat.ref})).
We consider three GENMETA estimators:
GENMETA.0, which is the initial GENMETA estimator with identity weighting matrix
and GENMETA.1 and GENMETA.2,
that use covariance estimates from the reference dataset and the studies, respectively.

From the results shown in Table \ref{table:logistic:bias se rmse cr al},
we observe that all three GENMETA estimators are nearly unbiased. The standard error estimates,
irrespective of whether $\Sigma_k,k=1,2,3$ were estimated using the study data sets or the reference sample,
accurately reflected the true standard errors of the GENMETA parameter estimates across different simulations.
As a result, the 95\% confidence intervals maintained the coverage probability at the nominal level.
Among the three GENMETA estimators considered,
clearly GENMETA.0, which use the non-optimal choice of $C=I$, is less efficient than GENMETA.1 and GENMETA.2,
which, between themselves, had comparable efficiency.

\begin{table}
\centering
 \begin{threeparttable}
\caption{Results on the GENMETA estimators}{
  \begin{tabular}{clrlcll}
 $n=50$   & $\beta_{i}^{*}$ & \multicolumn{1}{c}{Bias} & \multicolumn{1}{c}{SD ($\text{ESD}_1$, $\text{ESD}_2$)} & \multicolumn{1}{c}{RMSE} &
    \multicolumn{1}{c}{CR} & \multicolumn{1}{c}{AL} \\  
\vspace{-3mm} &&&&&& \\
       & $\beta_{1}^{*}$ &  .010  & .161 (.161, .162) & .161 & .968, .964 & .642, .636 \\
  GENMETA.0 & $\beta_{2}^{*}$ &  .005  & .110 (.111, .108) & .110 & .958, .960 & .434, .423 \\
       & $\beta_{3}^{*}$ & -.001  & .138 (.143, .142) & .138 & .963, .964 & .559, .556 \\ 
\vspace{-3mm} &&&&&& \\
       & $\beta_{1}^{*}$ &  .005  & .117 (.116, .110) & .117 & .976, .966 & .455, .433 \\
  GENMETA.1 & $\beta_{2}^{*}$ & -.003  & .101 (.105, .099) & .101 & .964, .955 & .411, .386 \\
       & $\beta_{3}^{*}$ &  .001  & .099 (.102, .097) & .099 & .973, .961 & .402, .381 \\ 
\vspace{-3mm} &&&&&& \\
       & $\beta_{1}^{*}$ &  .007  & .115 (.116, .111) & .115 & .971, .964 & .455, .435 \\
  GENMETA.2 & $\beta_{2}^{*}$ & -.003  & .102 (.105, .099) & .102 & .960, .959 & .413, .388 \\
       & $\beta_{3}^{*}$ &  .003  & .098 (.103, .098) & .098 & .957, .957 & .403, .383 \\ \hline
 \end{tabular}}
 \begin{tablenotes}[para,flushleft]
      \small
      \item Biases,
standard deviation (SD),
estimated standard deviation (ESD),
square roots of mean square errors (RMSE),
coverage rates (CR)
and average lengths (AL)
of 95\% confidence intervals for GENMETA.0 (the initial GENMETA estimator with identity weighting matrix),
GENMETA.1 and GENMETA.2 (the iterated GENMETA estimators without and with using the study covariance estimators) in the logistic regression setting.
Standard deviations were estimated either using the reference sample ($\text{ESD}_1$)
or using the covariance estimates of reduced model parameters from the studies ($\text{ESD}_2$).
Estimated standard deviations are reported by taking averages over simulated datasets.
Both estimated SE's are used to construct 95\% confidence intervals and their CR's and AL's are reported.
    \end{tablenotes}
    \end{threeparttable}
\label{table:logistic:bias se rmse cr al}
\end{table}

In the same setting as above,
when we vary $n$ from 10 up to the maximum of 1000 (Figure \ref{figure:rmse.GM.0 GM.1 GM.2: beta 1}),
we observe that the precision of the GENMETA estimates do not increase with $n$ once it reaches a threshold around 100, which is one third of the minimum of the study sample sizes($n_1=300$).
These thresholds were even smaller for estimation of coefficients associated with $X_2$,
which had weak to moderate correlation with the other covariates in the model.
The fact that the reference dataset can be substantially smaller than the study datasets without having much impact on the precision of the GENMETA estimator is encouraging given that accessing reference dataset of large sample size may be difficult in practice.

\begin{figure}
  \centering
  \includegraphics[scale=0.57]{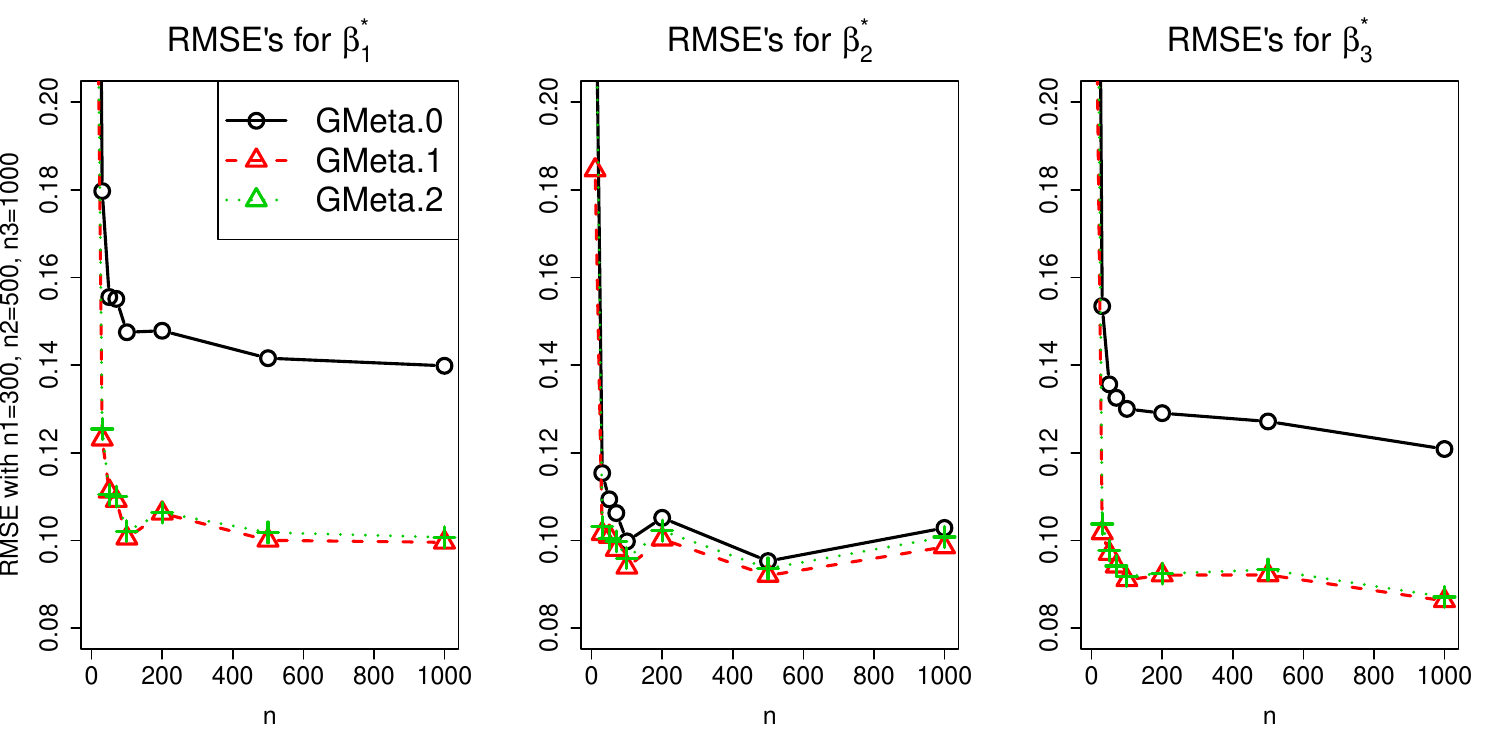}
  \caption{
  Square roots of mean square errors (RMSE) of GENMETA estimators
  for $\beta_{1}^{*}$, $\beta_{2}^{*}$ and $\beta_{3}^{*}$
  with fixed study sample sizes $n_1=300$, $n_2=500$ and $n_3=1000$
  and varying reference sample size $n$ from 10, 30, 50, 70, 100, 200 to 1000.
  The circle and solid line are for the RMSE's of GENMETA.0;
  the triangle and dashed line are for those of GENMETA.1;
  the plus and dotted line are for those of GENMETA.2.}
  \label{figure:rmse.GM.0 GM.1 GM.2: beta 1}
\end{figure}

Finally, we conduct additional simulation studies to obtain more insight into results from the real data analysis (Section \ref{sec:Real Data Analysis}).
Here, the settings are similar to before except we assume there are only two studies:
study-I fits the maximal logistic regression model involving all the three covariates
and study-II involves only two covariates, namely $X_1$ and $X_2$.
We assume 
$\rho_{I}=\rho_{II}=\rho_{b}$.
In our estimation,
we further considered an added complexity to account for study specific intercept terms for the maximal logistic regression model
\begin{equation*}
  Y\mid(X_{1}, X_{2}, X_{3}, \text{study}) \sim \text{Bernoulli}([1+\exp\{-(\beta_{0,\text{study}}^{*}+\beta_{1}^{*}X_{1} + \beta_{2}^{*}X_{2} + \beta_{3}^{*}X_{3})\}]^{-1})
\end{equation*}
so that the prevalence of the outcome, $\mbox{pr}(Y=1)$, could be different across the two studies.
In this setting,
the maximal set of parameters that are to be estimated through GENMETA can be defined as $\beta^{*}=(\beta_{0,\mbox{study-I}},\beta_{0,\mbox{study-II}},\beta_1,\beta_2,\beta_3)$.
We simulated data using values of intercept parameters that are identical across the two models,
but for estimation we allowed the intercept parameters to be different.
For the sake of comparison,  we also fitted a reduced model for study-I and conducted a standard multivariate meta-analysis of the underlying common parameters ($\theta_1$ and $\theta_2$) across the two studies. We assume the sample sizes for the two studies to be $n_{1}=500$ and $n_{2}=5000$, and that for the reference dataset to be $n=300$.

From the results reported in Table \ref{table:simu for data analysis},
we observe that in this simulation setting the reduced models produce biased estimate for $\beta_1^{*}$,
but not for $\beta_2^{*}$.
The result is intuitive given that the omitted covariate $X_3$ is primarily correlated with $X_1$.
As a result, standard meta-analysis was nearly unbiased for $\beta_2^*$, but not for $\beta_1^*$.
Parameter estimates from the maximal model from study-I are unbiased for all parameters,
but have much larger standard error compared to meta-analysis for estimation of $\beta_2^*$.
The GENMETA estimator produced unbiased estimates for all parameters and at the same time has comparable efficiency
as standard meta-analysis for estimation of $\beta_2^*$.
These results highlight the desirable feature of the GENMETA estimator that it can effectively combine information across studies to minimize bias due to omitted covariates and yet utilize all the information available across the partially informative studies.

\begin{table}
\centering
\begin{threeparttable}
\caption{A Simulation for Understanding Real Data Analysis}{%
    \begin{tabular}{lccccccc}
     & \multicolumn{2}{c}{Study I} & \multicolumn{1}{c}{Study II} & \multicolumn{1}{c}{Meta} & \multicolumn{2}{c}{GENMETA} \\ 
      & Maximal & Reduced & Reduced & Reduced  & Reduced & Maximal \\
  $\beta_i^*$  & PE (SD) & PE (SD) & PE (SD) &  PE (SD) & PE (SD) & PE (SD) \\ 
\vspace{-3mm} &&&&&&& \\
    $\beta_{1}^{*}$ & .270 (.149) & .429 (.116) & .424 (.037) &  .424 (.035) & .425 (.035) & .268 (.088) \\
    $\beta_{2}^{*}$ & .263 (.111) & .243 (.112) & .236 (.035) &  .236 (.034) & .237 (.034) & .263 (.039) \\
    $\beta_{3}^{*}$ & .258 (.136) & NA & NA  & NA & NA & .255 (.135) \\
    \hline
  \end{tabular}}
  \label{table:simu for data analysis}
  \begin{tablenotes}[para,flushleft]
      \small
      \item Point estimates (PE) and standard deviations (SD) from logistic regression with
   reduced and maximal models, meta-analysis and GENMETA estimation
   with $\beta_{1}^{*}=\beta_{2}^{*}=\beta_{3}^{*}=\log(1.3)\approx .262$.
   NA means there is no corresponding estimator.
    \end{tablenotes}
  \end{threeparttable}
\end{table}

\subsection{Heterogeneous Population}

In this section, we consider simulation studies where the underlying assumption of the homogeneity of covariate distribution across populations may be violated in multiple different ways. As a bench mark for comparison, we will describe setting (I) as the same setting as the the one we simulate under homogeneous population.
In the setting (II), we allow the means or/and variances to vary across the populations underlying the studies and reference sample, keeping the correlations to remain constant. Specifically, we assume the mean-vector for the three covariates can take one of three possible values:  $\mu_h = (1,1,1)$, $\mu_m = (0.5,0.5,0.5)$ and  $\mu_b = (0,0,0)$.  Similarly, the variance-vector is also allowed to vary across three possible set of values: $\sigma_h^2 = (2,2,2)$, $\sigma_l^2 = (0.5,0.5,0.5)$, $\sigma_b^2 = (1,1,1)$. In the setting (III),  we then allow the correlations among the covariates to vary across populations. Here we also allow three possible set of correlation vector $\rho$ as $\rho_{l}=(0.2,0.4,0.0)$ , $\rho_{h}=(0.4,0.8,0.2)$ and $\rho_{b}=0.3,0.6,0.1)$. Finally, we consider simulation setting (IV), where we allow for potential different inclusion criteria across studies leading to possible violations of the assumption of homogeneity of the covariate distribution. Specifically, we first simulate an underlying study base using the setup described in simulation setup (I), and then for study-I we only keep individuals with $X_1 >-0.5$ and $X_2 <0.5$, and in study-II we keep individuals with $X_1>0$. Finally, we consider an alternative simulation scenario where we assume the covariates are log-normally distributed by defining $X=\exp(W)$, where $W$ is generated from multivariate normal distribution following the same settings as I-IV described above

When covariates were normally distributed, we observe that (see Table \ref{table:robustness}) the proposed method is not very sensitive to underlying assumption of homogeneity of covariate distribution.  In the setting (II), where the mean or/and variances of the covariates are varied across the population, but correlations are kept fixed, there is virtually no bias. In setting (III), where correlations are varied,  we observe more noticeable, but still small, biases in parameter estimates. In setting (IV), when the inclusion criteria are varied across studies, there is also very minimal bias. When covariates are log-normally distributed, however, we observe that (see Table 1 in Supplementary Material) the method could be more sensitive to the violation of the underlying homogeneity assumption. In particular, when the inclusion criteria varied across studies (setting IV), large bias in point estimate and low coverage probability are observed for estimation of coefficient associated with $X_2$, the covariate which is used to define fairly non-overlapping inclusion criterion across two studies. Notably, even in this scenario, minimal bias is observed for estimation of the other covariates in the model.

\begin{table}
\begin{threeparttable}
\caption{Robustness of GENMETA Estimation (Normally Distributed Covariates)}{%
\begin{tabular}{ccccccrcccc}
  Setting  & Study-I & Study-II & Study-III & Reference & $\beta_{i}^{*}$ & Bias & SD (ESD) & RMSE & CR & AL \\ 
\vspace{-1mm} &&&&&&&&&& \\
  & $\mu_b$ & $\mu_b$ & $\mu_b$& $\mu_b$ & $\beta_{1}^{*}$ &  .001 & .111 (.112) & .111 & .947 & .437 \\
I  &$\sigma_b^2$ & $\sigma_b^2$ & $\sigma_b^2$ & $\sigma_b^2$ &  $\beta_{2}^{*}$ & -.002 & .098 (.099) & .098 & .956 & .389 \\
 & $\rho_b$ & $\rho_b$ & $\rho_b$ & $\rho_b$ & $\beta_{3}^{*}$ & .005 & .096 (.098) & .096 & .954 & .382 \\  
\vspace{-1mm} & & &&&&&&&& \\
  & $\mu_b$ & $\mu_h$ & $\mu_m$ &  $\mu_b$ & $\beta_{1}^{*}$ & .010 & .103 (.104) & .103 & .952 & .405 \\
  & $\sigma_b^2$  & $\sigma_b^2$ & $\sigma_b^2$ & $\sigma_b^2$ & $\beta_{2}^{*}$ &  -.006 & .083 (.083) & .083 & .954 & .324 \\
  & $\rho_b$ & $\rho_b$ & $\rho_b$ & $\rho_b$ & $\beta_{3}^{*}$ &  .005 & .085 (.088) & .085 & .956 & .343 \\   
\vspace{-3mm} & &  &&&&&&&& \\
    & $\mu_b$ & $\mu_b$ & $\mu_b$ & $\mu_b$ & $\beta_{1}^{*}$ &  .003 & .139 (.136) & .139 & .939 & .529 \\
 II   & $\sigma_b^2$ & $\sigma_h^2$ & $\sigma_l^2$ & $\sigma_b^2$ &  $\beta_{2}^{*}$ &  -.003 & .084 (.086) & .084 & .956 & .335 \\
  & $\rho_b$ & $\rho_b$ & $\rho_b$ & $\rho_b$ & $\beta_{3}^{*}$ &  .003 & .112 (.111) & .112 & .949 & .431 \\   
\vspace{-3mm} &&&&&&&&&&\\
   & $\mu_b$ & $\mu_h$ & $\mu_m$ &  $\mu_b$  & $\beta_{1}^{*}$ & .013 & .124 (.126) & .125 & .946 & .493 \\
  & $\sigma_b^2$ & $\sigma_h^2$ & $\sigma_l^2$ & $\sigma_b^2$ & $\beta_{2}^{*}$&  -.006 & .073 (.075) & .073 & .958 & .291 \\
   & $\rho_b$ & $\rho_b$ & $\rho_b$ & $\rho_b$ & $\beta_{3}^{*}$ &  .005 & .097 (.100) & .097 & .949 & .391 \\   
\vspace{-1mm} & & &&&&&&&& \\
  & $\mu_b$ & $\mu_b$ & $\mu_b$ & $\mu_b$ & $\beta_{1}^{*}$ & -.092 & .142 (.151) & .169 & .958 & .579 \\
 & $\sigma_b^2$ & $\sigma_b^2$ & $\sigma_b^2$ & $\sigma_b^2$  & $\beta_{2}^{*}$&   .019 & .105 (.109) & .107 & .963 & .423 \\
  & $\rho_b$ & $\rho_b$ & $\rho_b$ & $\rho_h$ & $\beta_{3}^{*}$ &  .053 & .120 (.129) & .131 & .971 & .495 \\   
\vspace{-3mm} & & &&&&&&&& \\
  &  $\mu_b$ & $\mu_b$ & $\mu_b$ & $\mu_b$ & $\beta_{1}^{*}$ &  .035 & .099 (.099) & .106 & .917 & .385 \\
 &  $\sigma_b^2$ & $\sigma_b^2$ & $\sigma_b^2$ & $\sigma_b^2$ &  $\beta_{2}^{*}$&  .002 & .096 (.096) & .096 & .954 & .377 \\
  & $\rho_b$ & $\rho_b$ & $\rho_b$ & $\rho_l$ & $\beta_{3}^{*}$ &  .012 & .087 (.087) & .088 & .944 & .343 \\   
  \vspace{-3mm} & & &&&&&&&& \\
   &  $\mu_b$ & $\mu_b$ & $\mu_b$ & $\mu_b$ & $\beta_{1}^{*}$ &  .060 & .113 (.113) & .128 & .916 & .443 \\
 III &  $\sigma_b^2$ & $\sigma_b^2$ & $\sigma_b^2$ & $\sigma_b^2$ &  $\beta_{2}^{*}$&  -.001 & .096 (.097) & .096 & .955 & .379 \\
  & $\rho_l$ & $\rho_b$ & $\rho_h$ & $\rho_l$ & $\beta_{3}^{*}$ &  -.006 & .103 (.102) & .104 & .944 & .398 \\   
  \vspace{-3mm}  & & &&&&&&&& \\
   &  $\mu_b$ & $\mu_b$ & $\mu_b$ & $\mu_b$ & $\beta_{1}^{*}$ &  .039 & .130 (.132) & .135 & .939 & .515  \\
   &  $\sigma_b^2$ & $\sigma_b^2$ & $\sigma_b^2$ & $\sigma_b^2$ &  $\beta_{2}^{*}$&  -.006 & .097 (.100) & .097 & .958 & .392 \\
  & $\rho_l$ & $\rho_b$ & $\rho_h$ & $\rho_b$ & $\beta_{3}^{*}$ &  -.027 & .116 (.118) & .119 & .944 & .461 \\   
  \vspace{-3mm}  & & &&&&&&&& \\
   &  $\mu_b$ & $\mu_b$ & $\mu_b$ & $\mu_b$ & $\beta_{1}^{*}$ &  -.036 & .165 (.173) & .169  & .957 & .671 \\
  &  $\sigma_b^2$ & $\sigma_b^2$ & $\sigma_b^2$ & $\sigma_b^2$ &  $\beta_{2}^{*}$&  .013 & .103 (.109) & .104 & .962 & .424 \\
& $\rho_l$ & $\rho_b$ & $\rho_h$ & $\rho_h$ & $\beta_{3}^{*}$ &  .003 & .143 (.153) & .143 & .959 & .591 \\   
  \vspace{-1mm}  & & &&&&&&&& \\
 &  & & $\mu_b$ & $\mu_b$ & $\beta_{1}^{*}$ &  .014 & .123 (.127) & .124  & .961 & .494 \\
 IV &  $X_1>-0.5,$ & $X_2 > 0$ & $\sigma_b^2$ & $\sigma_b^2$ &  $\beta_{2}^{*}$&  -.008 & .105 (.109) & .105 & .965 & .428 \\
  &  $X_2<0.5$  &  & $\rho_b$ & $\rho_b$ & $\beta_{3}^{*}$ &  -.001 & .094 (.093) & .093 & .958 & .366 \\  
 \hline
\end{tabular}}
\label{table:robustness}
\begin{tablenotes}[para,flushleft]
      \small
      \item Biases,
standard deviation (SD),
estimated standard deviation (ESD),
square roots of mean square errors (RMSE),
coverage rates (CR),
and average lengths (AL)
of 95\% confidence intervals of the GENMETA estimates using the study covariance estimators in the setting of logistic regression. In setting (I), data are simulated in ideal setting there the covariate distribution, characterized by mean, sd and correlation of normal variates, are assumed to same across all populations. In setting (II)-(IV), the assumption is violated by creating variations in mean/sd, correlations and selection criterion across the studies and reference sample. with different study and reference sample. The vector of covariate means, variances and correlations are denoted by denoted by $\mu_*  = (\mu_1,\mu_2,\mu_3)$, $\sigma^2_*  = (\sigma_1^2,\sigma_2^2, \sigma_3^2)$ and $\rho_*  = (\rho_{12},\rho_{23},\rho_{13})$ for $* \in \{b,l,m,h\}$, where $\mu_b = (0,0,0)$, $\mu_m = (0.5,0.5,0.5)$, $\mu_h = (1,1,1)$; $\sigma_b^2 = (1,1,1)$, $\sigma_l^2 = (0.5,0.5,0.5)$, $\sigma_h^2 = (2,2,2)$ and $\rho_b = (0.3,0.6,0.1)$, $\rho_h = (0.4,0.8,0.2)$, $\rho_l = (0.2,0.4,0)$. Estimated standard deviation are obtained by the asymptotic formula (\ref{equ:covariance matrix}) and used to construct 95\% confidence interval.
    \end{tablenotes}
  \end{threeparttable}
\end{table}
\subsection{Power Evaluation of the Diagnostic Test ($T_{GENMETA}$)}
We assess the power of the proposed test statistic, $T_{GENMETA}$ in the presence of heterogeneity in the regression parameters ($\beta$) across the studies. In the context of standard multivariate meta-analysis, where it is assumed that all the studies ascertain the same set of covariates, test for heterogeneity is performed using standard multivariate Cochran's test-statistic in the form 
\begin{equation*}
Q = \sum_{k=1}^K(\hat{\beta}_k - \hat{\beta}_{meta})^TS_k^{-1}(\hat{\beta}_k - \hat{\beta}_{meta})
\end{equation*}
where $\hat{\beta}_{meta} $ is the usual multivariate meta-analysis estimate and $S_k$ is the standard error of $\hat{\beta}_k$ for $k = 1,\dots,K$. We will utilize $Q$ as a benchmark to evaluate the power of $T_{GENMETA}$. 

In all simulations,  as before, we assume the existence of three separate studies and relationship between a binary outcome variable $Y$ and three covariates $(X_1,X_2,X_3)$ in each study follows the same logistic regression model of the form (\ref{Equation:Simu}).
However, instead of assuming a fixed set of $\beta$ across all studies, we simulate different values of $\beta$ from a normal distribution with mean $(\beta_0^*, \beta_1^*, \beta_2^*, \beta_3^*) = (\log 1.3, \log 1.3, \log 1.3)$ and variance $\sigma^2I$, where the parameter $\sigma^2>0$ is varied to control the degree of heterogeneity across studies.
As before, we assume that  
$(X_1,X_2,X_3)$ follows a multivariate normal distribution with mean zero, unit variances and underlying correlations
$\rho=\rho_{12}=0.3,\rho_{13}=0.6,\rho_{23}=0.1)$ across all the three studies.
We simulate data for the different studies from the above random-effects logistic regression model and then fit reduced models of the form (\ref{Equation:Simureduced})
to the three different studies. In particular, we assume $X_1$ and $X_2$ included in Study-I, $X_2$ and $X_3$ in Study-II and $X_1$ and $X_3$ in Study-III. We fix the sample size of the studies at $n_1=3000$, $n_2=5000$ and $n_3=10000$ and vary sample size of the reference dataset. The level of the test is set to 5\%.  For the purpose of comparison, we also fit the maximal model to each study involving all three covariates and apply the standard Q-statistics for testing heterogeneity. 

Comparison of power of $T_{GENMETA}$ and $Q$ statistics shows that, as expected, the power for both tests increases as a function of degree of heterogeneity, $\sigma^2$ (Figure \ref{figure:Power_n_ref_50_simulation_setting_1}). Clearly, $T_{GENMETA}$ suffers some loss of power as it handles missing covariates, but it retains substantial power, even with small reference dataset ($n=100$), to remain practically useful.
\begin{figure}
  \centering
  \includegraphics[width=12cm,height=10cm, keepaspectratio]{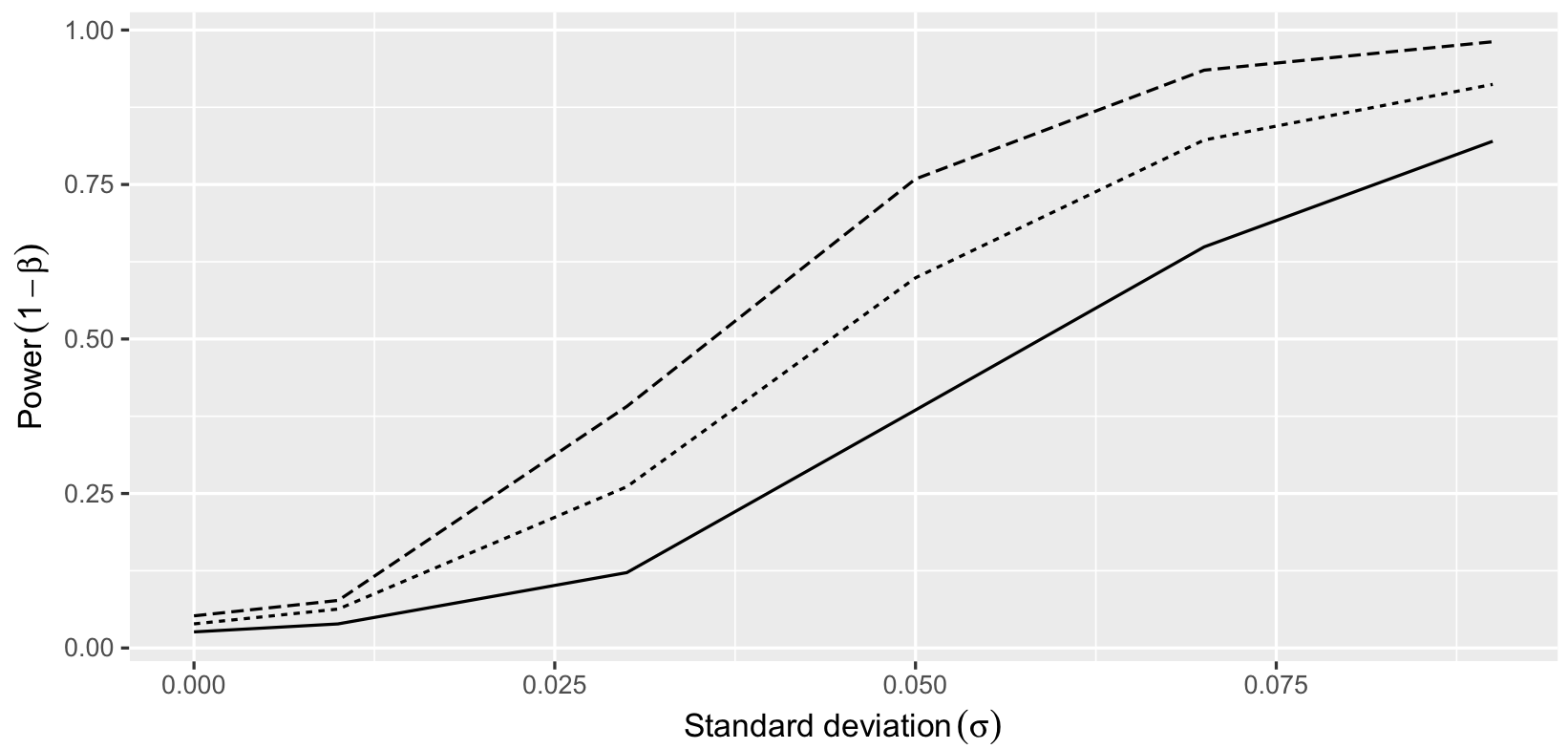}
  \caption{Power curves of simple multivariate meta-analysis test statistic ($Q$) and $T_{GENMETA}$ for simulated datasets. The long-dashed line is for the simple meta-analysis estimator. The solid and dotted lines are for GENMETA estimators with reference data sample sizes 100 and 500, respectively. Level of the test ($\alpha$) is set to 0.05.}
  \label{figure:Power_n_ref_50_simulation_setting_1}
\end{figure}

\section{Real Data Analysis}
\label{sec:Real Data Analysis}
In this section,
we illustrate an application of the proposed methodology to develop a model for predicting risk of breast cancer based on combination of different risk factors using data from multiple studies.
The first study, the Breast Prostate Colorectal Cancer Cohort study (BPC3),
includes a total of 7448 cases and 8812 controls, drawn from eight different underlying cohorts.
Details of the study, including its recent application for the development of breast cancer risk prediction model,
can be found elsewhere \citep{mass}.
In the current analysis, we focus on the analysis of breast cancer risk associated with a selected set of factors,
including family history, age at menarche, age at first birth and weight.
The second study involves a dataset involving 1217 cases and 1616 controls from the Breast Cancer Detection and Demonstration Project (BCDDP).
The study has been previously used to develop an updated version of the widely popular Breast Cancer Risk Assessment tool \citep{chen} to incorporate mammographic density,
the areal proportion of breast tissue that is radiographically dense,
known to be a strong risk factor for breast cancer.
The dataset from the BCDDP study included mammographic density and number of previous breast biopsy,
in addition to all the factors considered in the BPC3 data analysis.
Let $X$ denote the common set of covariates that are measured across both the studies and $Z$ be the factors that are available only in BCDDP.
The goal is to estimate parameters associated with an underlying logistic regression model
that includes all of the different factors.
While the BPC3 study is large in size and represents multiple populations,
it has information on more limited number of risk factors.
The BCDDP study, on the other hand, has information on extended set of risk factors,
but is much smaller in size.
A combined analysis of these two studies can potentially lead to more generalizable and precise estimate of risk parameters.

Throughout the analysis, we used a sample of 137 cases and 163 controls from the BCDDP study
as the reference sample based on which the distribution of covariates are estimated.
To maintain independence of the reference and study samples, we exclude the reference sample from the primary analysis of the BCDDP study that involved estimation of the log-odds-ratio parameters. Further, both the studies involve case-control sampling with similar case-control proportions. In general, if non-random sampling is used for selection of subjects in any of the studies, then the covariate distribution underlying the GENMETA estimating equation needs to be adjusted to account for the study design. In this application, because we had access to the the BCDDP study, we could adjust for the design effect by simply selecting a reference sample that includes cases and controls in similar ratio as the main studies. In general, however, the effect of non-random sampling design for the main studies may need to adjusted through careful weighting of subjects in the reference sample.

For each of the eight cohorts within the BPC3 study and for the BCCDP study,
we first fit a reduced logistic regression model including $X$. 
All models included age as an additional cofactor and included study specific intercept parameters and age effects.
Specifically, we consider underlying models in the form
\begin{equation}\label{eq:BPC3Model}
(Y\mid X, \text{Age},  \text{study}=k)
\sim \text{Bernoulli}((1+\exp\{-(\theta_{0k} + \theta_{A_k}Age+\bd{\theta}^T_XX )\})^{-1}).
\end{equation}

We applied the diagnostic test for model violation  to these datasets. We found the value of the test-statistic ($\hat{T}_{GENMETA}$) to be 59.01 and the corresponding p-value to be 0.366 under a $\chi^2_{(56)}$ distribution. Thus, it appears that the underlying model assumptions are unlikely to be grossly violated in this application.

First, to illustrate how the proposed GENMETA estimator compares to standard meta-analysis method,
we consider estimating the common underlying parameters of interest $\bd{\theta}_{X}$ using these two alternative methods.
We fitted model (\ref{eq:BPC3Model}) separately for each study and obtained estimates of the parameters and covariance matrices.
Then, for the underlying common parameter of interests $\bd{\theta}_{X}$,
we conducted a standard multivariate meta-analysis using the corresponding subset of parameters estimates and covariance matrices.
Alternatively, using the parameters estimates and variance-covariance matrices from the individual studies, and using the set aside BCDDP sample as the reference dataset to estimate the joint distribution of $X$ and $age$, we estimated all of the parameters of model (\ref{eq:BPC3Model}) using the GENMETA procedure.
From the results reported in Table \ref{table:Analysis.24.25},
we observe that in this setting,
the meta-analysis and GENMETA estimators produce similar estimates
as well as their standard errors across all the different risk-factors of interest.
In one of the results stated earlier, we have seen theoretically that in an idealized setting
where all the models and underlying populations are identical,
the two estimators are asymptotically equivalent.
It's encouraging to observe the close correspondence
between the estimators in the data analysis,
which includes a diverse set of studies that are likely to have significant heterogeneity across the underlying populations.
In particular, for a number of the risk-factors (e.g family history),
coefficient estimates were noticeably different across the two studies.
When significant heterogeneity existed,
the meta-analyzed estimates were pooled closer to those from the BPC3 study due to its large sample size.

Next, we turn our attention to the analysis of data from the BCDDP study
using a maximal model that includes $X$ and the additional covariates, mammographic density and number of previous breast biopsy.
Comparison of the parameter estimates associated with
$X$ across the maximal and reduced model within the BCDDP study
indicates major differences in the estimates of the coefficients associated with weight.
In the maximal model, higher weight is found to be be much more strongly associated with increased risk of breast cancer.
The unmasking of the effect of weight in the maximal model is intuitive
given that body weight and mammographic density is known to have strong negative correlation.
Although not as dramatic,
there are some differences in effects of age at menarchy and age at first birth between the maximal and reduced models,
also possibly because of modest correlation of these factors with mammographic density and number of previous breast biopsy.
The effect of family history, however, is almost identical across the two models.

Finally, we used the GENMETA method to combine estimates of the parameters of the maximal model from the BCDDP study
and those from the reduced models from the eight BPC3 cohorts.
We assumed an underlying maximal model of interest across the 9 studies in the form
\begin{equation*} \label{eq:DataFullModel}
(Y\mid X, Z, \mbox{Age},  \mbox{study=k})
\sim
\text{Bernoulli}
([1+\exp\{-(\theta_{0k} + \theta_{A_k}Age+ \bd{\beta}^T_{X}X + \bd{\beta}^T_{Z}Z)\}]^{-1}).
\end{equation*}

We observe that GENMETA produces estimates of effect of family history and associated standard error very similar to those observed based on the standard meta-analysis of the reduced models across the nine cohorts. The estimate is pooled heavily towards the BPC3 study due to its large sample size.
In contrast, the GENMETA estimates for weight are very similar to those
observed from the maximal model only within the BCDDP study.
These results are consistent with simulation studies,
where GENMETA behaves similar to reduced model meta-analysis
when omitted covariates do not cause notable bias.
In contrast,
when omitted covariates cause important bias,
the GENMETA estimator is pooled towards estimates from maximal or more complete models that may be available
from a restricted set of studies.
The behavior of GENMETA for the two other covariates, age at menarchy and age at first birth, were in between, which is also intuitive
given that we had observed their coefficients changed notably,
but less dramatically,
in the maximal model compared to the reduced model within the BCDDP study.
The GENMETA parameter estimates and standard errors for the additional variables mammographic density and number of previous breast biopsy,
were similar to those observed for the maximal model in the BCDDP,
the only study which had information on these two factors.
Thus, overall the data analysis illustrates that the GENMETA estimator behaves in a similar manner as meta-analysis for combining information across multiple possibly heterogeneous studies, but it has the added flexibility to effectively combine information from disparate models.

\section{Discussion}
The proposed method can be viewed as a natural extension of the traditional fixed effect meta-analysis method
that is widely used in practice.
Both simulation studies and data analysis demonstrate that the method not only provides theoretically valid and efficient inference
in idealized conditions, but also can perform robustly in non-idealized settings. 
A critical element of the proposed method is the access to a reference dataset.
While the ideal choice of the reference dataset will vary by applications,
publicly available survey data,
which collect information on a wide variety of factors,
can be useful broadly.
In fact, in large scale genetic association studies,
use of reference samples, such as the 1000 Genome study,
are commonly used for estimation correlation parameters
across genetic markers in the genome \citep{1kgenome, 1kgenome2, lee}.
For epidemiologic studies,
good resources for reference dataset for the US population include
the National Health Interview Survey \citep{adams, botman,bloom} and the National Health and Nutrional Examination Survey \citep{fang, jiang, ferranti, idler, goodman}, which routinely collect data on a wide variety of health and lifestyle related factors. If multiple studies coordinate through consortium effort,
which is increasingly common in biomedical applications,
then studies which have most complete information, at least on some sub-samples, can provide reference sample.

When information on all covariates are not available in a single reference sample,
one may have to consider simulation for generating such data
by combining information from multiple studies under some modeling assumptions. As the access to large reference dataset
that is ideally representative of the underlying study populations can be difficult,
we found two aspects of GENMETA to  be appealing.
First, the sample size for the reference dataset can be small relative to the study datasets
and yet GENMETA can have reasonable efficiency.
In fact, increasing the sample size for the reference dataset beyond certain threshold
does not have an impact on the efficiency of GENMETA.
Second, although technically the method requires all the populations underlying the studies and the reference dataset to be the same,
in practice, the method can be robust to a reasonable degree of heterogeneity in distribution of covariates. However, it is possible to have a large bias when estimating coefficients associated with covariates that have been used to define widely varying inclusion criteria. When different studies follow very different designs it is best to obtain study-specific reference samples for estimating the underlying moment equations. Alternatively, it may be possible to modify a large reference sample by using study-specific sampling weights/inclusion criteria when estimating the moment equations. Dealing with study-specific covariates, such as centers within a study, can also pose challenges as information on such variables are not expected to be available from a common reference sample.  We have illustrated in our data example that it is possible to deal with such variables by imposing additional independence assumptions from other factors. In general, such complications need to be dealt in a case-by-case basis and some study specific reference samples may be needed to avoid making strong assumptions. Further research is merited to explore these and other practical challenges in implementation of the proposed method.

In general, we believe caution is needed for interpretations and applications of models
that may be developed by combining information from disparate models across multiple studies.
A model developed from a single study with complete information,
although may be inefficient and may lack generalizability,
is more likely to be internally consistent
and thus can provide valid etiologic inference even if it is not representative of the general population.
On the other hand,
etiologic interpretation of parameters can be difficult
when the underlying model is developed using information
across multiple studies
that are potentially heterogeneous.
For the development of predictive models, however,
where the focus is not so much parameter interpretation,
development of rich models by combining information across multiple studies
and then validating such models in independent studies can be an appealing strategy. These and other practical issues related to model development using multiple data sources have been also discussed in several recent articles \citep{wang_merging_2015, hanpeisong2017, cheng_informing_nodate, estes_empirical_2018, cheng_informing_nodate}.

In this article, we used generalized method of moments as the underlying inferential framework. Alternatively, inference could be also performed using empirical likelihood theory \citep{Qin1994, Qin2000, chatterjee} exploiting the same set of moment equations as we propose. While in small sample, empirical likelihood estimators may perform better,  implementation can be substantially more complex.  Recently,  a simulation based method has been also described for combining information on model  parameters across disparate studies \citep{Rahmandad2017}. 
Computationally, the proposed method may also enjoy substantial advantages in dealing with complex models, such as those in high-dimensional settings, where repeated model fitting on simulated data is extensive.
Further research is merited in multiple directions to increase the practical utility of GENMETA.
It is possible that in some applications we may have information only on subsets of parameters
underlying the fitted reduced models.
It's an open question how such partial information can be used to set up the underlying moment equations in the GENMETA procedure. Ideally, to increase robustness of inference, the GENMETA procedure should use study specific reference sample
for setting up the moment equations.
For this purpose, it may be useful to develop strategies to combine information
on a common reference sample with complete covariate information and data
from individual studies that have partial covariate information.

\begin{landscape}
\begin{table}[!ht]
\begin{threeparttable}
\caption{Real Data Analysis Results with Meta-Analysis and GENMETA Method}{%
  \begin{tabular}{lrrrrrrrrrrrrr}
     & \multicolumn{2}{c}{BCDDP} & \multicolumn{8}{c}{BPC3} & \multicolumn{1}{c}{Meta} & \multicolumn{2}{c}{GENMETA} \\ 
     & \multicolumn{1}{c}{Maximal}
     & \multicolumn{1}{c}{Reduced}
     & \multicolumn{1}{c}{CPS2}
     & \multicolumn{1}{c}{EPIC}
     & \multicolumn{1}{c}{MCCS}
     & \multicolumn{1}{c}{MEC}
     & \multicolumn{1}{c}{NHS}
     & \multicolumn{1}{c}{PLCO}
     & \multicolumn{1}{c}{WHI}
     & \multicolumn{1}{c}{WHS}
     & \multicolumn{1}{c}{Reduced}
     & \multicolumn{1}{c}{Reduced}
     & \multicolumn{1}{c}{Maximal}   \\
     & \multicolumn{1}{c}{Model}
     & \multicolumn{1}{c}{Model}
     & \multicolumn{1}{c}{Cohort}
     & \multicolumn{1}{c}{Cohort}
     & \multicolumn{1}{c}{Cohort}
     & \multicolumn{1}{c}{Cohort}
     & \multicolumn{1}{c}{Cohort}
     & \multicolumn{1}{c}{Cohort}
     & \multicolumn{1}{c}{Cohort}
     & \multicolumn{1}{c}{Cohort}
     & \multicolumn{1}{c}{Model}
     & \multicolumn{1}{c}{Model}
     & \multicolumn{1}{c}{Model}   \\
     \multicolumn{1}{c}{Risk Factors}
     & \multicolumn{1}{c}{PE(SE)}
     & \multicolumn{1}{c}{PE(SE)}
     & \multicolumn{1}{c}{PE(SE)}
     & \multicolumn{1}{c}{PE(SE)}
     & \multicolumn{1}{c}{PE(SE)}
     & \multicolumn{1}{c}{PE(SE)}
     & \multicolumn{1}{c}{PE(SE)}
     & \multicolumn{1}{c}{PE(SE)}
     & \multicolumn{1}{c}{PE(SE)}
     & \multicolumn{1}{c}{PE(SE)}
     & \multicolumn{1}{c}{PE(SE)}
     & \multicolumn{1}{c}{PE(SE)}
     & \multicolumn{1}{c}{PE(SE)}   \\ 
    FH1 & .80(.14) & .80(.14) & .47(.13) & .29(.15) & .56(.19) & .41(.28) & .48(.08) & .39(.13) & .30(.06) & .28(.19) & .40(.04) & .42(.04) & .37(.08) \\
    AMEN1 & .11(.10) & .07(.10) & -.03(.14) & .02(.09) & -.19(.17) & -.09(.24) & .06(.09) & -.05(.12) & .13(.08) & .03(.17) & .04(.04) & .03(.04) & .04(.06) \\
    AMEN2 & .55(.15) & .45(.15) & -.09(.17) & .04(.12) & -.44(.23) & .35(.35) & .19(.10) & .03(.15) & .19(.09) & .14(.19) & .13(.05) & .13(.05) & .32(.08) \\
    AFB1 & .06(.14) & .18(.15) & .28(.17) & .12(.14) & -.08(.25) & .06(.17) & .39(.20) & .16(.14) & .19(.09) & .92(.23) & .21(.05) & .20(.05) & .05(.09) \\
    AFB2 & .29(.20) & .46(.20) & .73(.24) & .24(.17) & .35(.30) & .05(.26) & .36(.22) & .52(.22) & .44(.13) & .96(.28) & .38(.06) & .38(.07) & .21(.12) \\
    WT1 & .29(.11) & .09(.11) & .09(.14) & -.01(.09) & .22(.18) & .09(.17) & .21(.08) & .09(.13) & -.03(.08) & -.01(.14) & .08(.04) & .08(.04) & .31(.07) \\
    WT2 & .52(.13) & .10(.13) & .16(.14) & .24(.11) & .45(.19) & -.08(.18) & .10(.08) & .09(.13) & .18(.08) & -.16(.15) & .14(.04) & .14(.04) & .63(.09) \\
    NBIOPS & .13(.09) & NA & NA & NA & NA & NA & NA & NA & NA & NA & NA & NA & .13(.10) \\
    MD & .46(.05) & NA & NA & NA & NA & NA & NA & NA & NA & NA & NA & NA & .43(.06) \\
   \hline
  \end{tabular}}
    \label{table:Analysis.24.25}
  \begin{tablenotes}[para,flushleft]
      \small
      \item Combined analysis of BCDDP and BPC3 study to develop  a multivariate logistic regression model for breast cancer risk.  For each cohort within BPC3 and for BCDDP,  standard logistic regression model is applied for fitting reduced models including  FH (family history), AMEN (age at menarche), AFB (age at first live birth) and WT (weight). Parameter estimates of the reduced models across studies are then combined using standard meta-analysis (meta) or GENMETA. For the BCDDP study, a maximal logistic model is fitted including additional covariates mammographic density (MD) and number of previous biopsy (NBIOPS).
These estimates are then combined with with estimates of reduced model parameters from BPC3 studies to obtain GENMETA estimates of the maximal model.  Point estimates (PE) and standard errors (SE) are shown for each analysis.
NA means there is no corresponding estimator.
The variables analyzed include:
FH: binary indicator of family history;
AMEN1 and AMEN2: dummy variables associated with age-at-menarche categories $\geq 14$, $12$--$13$ and $\leq 11$;
AFB1 and AFB2: dummy variables associated with age-at-first-live-birth categories $\leq 20$, $21$--$29$ and $\geq 30$;
WT1 and WT2: dummy variables associated with weight categories $\leq 62.6$, $62.6$--$73.1$ and $\geq 73.1$ in kilograms;
NBIOPS: the number of biopsies coded as a continuous variable
and MD: the standardized mammographic density coded as a continuous variable.
BPC3 contains eight cohorts with abbreviated names CPS2, EPIC, MCCS, MEC, NHS, PLCO, WHI and WHS.
    \end{tablenotes}
  \end{threeparttable}
 \end{table}
\end{landscape}

\section*{Acknowledgement}
Research reported in this article was partially funded through a Patient-Centered Outcomes Research Institute (PCORI) Award (ME-1602-34530) and funding by NIH for the Environmental influences of Child Health Outcomes (ECHO) Cohort Data Analysis Center (U24OD023382). The statements, opinions in this article are solely the responsibility of the authors and do not necessarily represent the views of the Patient-Centered Outcomes Research Institute (PCORI), its Board of Governors or Methodology Committee.

\section*{Software}
\label{sec5}
A software in the form of an R Package (GENMETA) is available on Github through the link : \textit{https://github.com/28pro92/packages-GENMETA}.

\section*{Supplementary Material}
\label{SM}
Supplementary material is available online which includes proof of theorem 1, all the derivations and a table containing the simulation results for log-normally distributed covariates .

\section*{Appendix}
\subsection*{ Assumptions of Theorem 1}
Assumptions (A1)-(A4) are for consistency and
the additional assumptions (A5)-(A9) are for asymptotic normality.

(A1): $C$ is positive semi-definite
and $C E \{U(X; \bd{\beta}, \bd{\theta}^{*})\}=0$ if and only if $\bd{\beta}=\bd{\beta}^{*}$.

(A2): $\bd{\beta}^{*}\in D_{\bd{\beta}}$, which is compact.

(A3): $u_{k}(X;\bd{\beta},\bd{\theta}_{k})$ is continuous for each $(\bd{\beta}, \bd{\theta}_{k}) \in D_{\bd{\beta}} \times \mathcal{N}(\bd{\theta}_{k}^{*})$ with probability one, where $\mathcal{N}(\bd{\theta}^{*}_{k})$ is a neighborhood of $\bd{\theta}^{*}_{k}$ for $k = 1,\dots,K$.

(A4): $E\{\sup_{(\bd{\beta}, \bd{\theta}_{k})\in D_{\bd{\beta}} \times \mathcal{N}(\bd{\theta}^{*}_{k})}||u_{k}(X;\bd{\beta}, \bd{\theta}_{k})||\} < \infty$ for $k = 1,\dots,K$.

(A5): ${\partial}u_{k}(X;\bd{\beta}, \bd{\theta}_{k})/{\partial \bd{\beta}}$ is continuous at each $(\bd{\beta}, \bd{\theta}_{k}) \in \mathcal{N}(\bd{\beta}^{*})\times \mathcal{N}(\bd{\theta}_{k}^{*})$ with probability 1,
where $N(\beta^{*})$ is a neighborhood of $\beta^{*}$.

(A6): $E\{\sup_{(\bd{\beta}, \bd{\theta}_{k})\in \mathcal{N}(\bd{\beta}^{*})\times \mathcal{N}(\bd{\theta}^{*}_{k})}
||{\partial}u_{k}(X, \bd{\beta}, \bd{\theta}_{k})/{\partial \bd{\beta}}||\} <\infty$.

(A7): ${\partial}u_{k}(X;\bd{\beta}^{*}, \bd{\theta}_{k})/{\partial \bd{\theta}_{k}}$ is continuous at
each $\bd{\theta}_{k} \in \mathcal{N}(\bd{\theta}_{k}^{*})$ with probability one.

(A8): $E\{\sup_{\bd{\theta}_{k}\in \mathcal{N}(\bd{\theta}^{*}_{k})}
||{\partial}u_{k}(X, \bd{\beta}^{*}, \bd{\theta}_{k})||/{\partial \bd{\theta}_{k}}\} <\infty$.

(A9): $\bd{\Delta}(\bd{\beta}^{*}, \bd{\theta}^{*})$ exists and is finite and $\bd{\Gamma}(\bd{\beta}^{*}, \bd{\theta}^{*})$ is of full rank.

\subsection*{More on the Global Identification Assumption (A1)} Sometimes it's difficult to practically check the global identification condition. This motivates us to investigate conditions for local identifiability, or equivalently,  the invertibility of the matrix of second derivatives at the true parameter, i.e.,
${\partial ^2}Q(\bd{\beta})/{\partial \bd{\beta}^2}\mid_{\bd{\beta} = \bd{\beta^*}}
=
[E\{{\partial}U(X;\bd{\beta})/{\partial \bd{\beta}}\}^TC
E\{{\partial}U(X;\bd{\beta})/{\partial \bd{\beta}}\}]\mid_{\bd{\beta} = \bd{\beta^*}}$  (Rothenberg, 1971; Engle
and McFadden, 1994), assuming $C$ is a positive definite matrix. The condition can be stated in terms of the equivalent sample version of the matrix, given by, $X^T_{rbind}WX_{A_{diag}}CX^T_{A_{diag}}WX_{rbind}$. As $C$ is a positive definite matrix, the entire local identifiability condition for the sample version then boils down to $X^T_{A_{diag}}WX_{rbind}$ being a full column rank matrix. A sufficient condition for this is $X_{A_{diag}}$ contains information on all the covariates of the maximal model. In other words, the individual covariates in the maximal model have to be part of at least one of the reduced models.
\appendix

\bibliographystyle{authordate1}
\bibliography{GENMETA_arxiv}

\section{Supplementary}
\beginsupplement
\subsection{Asymptotic Equivalence of GENMETA Estimator and Simple Meta-Analysis Estimator
When All the Reduced Models Are the Same to the Maximal Model}
When all the reduced models are the same to the maximal model,
it follows
$\bd{\theta}^*_{k} = \bd{\beta}^*$,
$X_{A_k} = X$
and $g_{k}=f$ for $k=1,2, \ldots, K$.
Then, for each $k$,
$u_{k}(\bd{X}; \bd{\beta}^*,\bd{\theta}^*_{k})
= u_{k}(\bd{X}; \bd{\beta}^*,\bd{\beta}^{*})
= \int s_{k}(y\mid X_{A_{k}};\bd{\beta}^*)f(y\mid X;\bd{\beta}^*)dy = \vec{0}$.
By the definition of $\bd{\Delta}$, we have $\bd{\Delta} = \vec{0}$.
On the other hand,
assuming
$E_{Y\mid\bd{X}}\{\nabla_{\bd{\theta}_{k}} s_{k}(\bd{\theta}_{k}^{*})\}
=
\nabla_{\bd{\theta}_{k}}E_{Y\mid\bd{X}}\{s_{k}(\bd{\theta}_{k}^{*})\}$
with $s_{k}(\bd{\theta}_{k}^{*})=s_{k}(Y\mid X_{A_k}; \bd{\theta}_{k}^{*})$,
it follows
$\bd{\Lambda}_{k}=(1/c_{k})I(\bd{\theta}_{k}^{*})$,
where $I(\bd{\theta}_{k}^{*})$ is the Fisher's information matrix of $g_{k}$ or $f$.
Then, the optimal $C$ is
\begin{equation*}
C_{\text{opt}}
= \bd{\Lambda}^{-1}
= \diag(c_1\bd{\Sigma},\dots,c_K\bd{\Sigma}),
\end{equation*}
where $\bd{\Sigma}=I(\bd{\theta}_{k}^{*})^{-1}$.
Denote as $\hat{C}_{\text{opt}}$ a consistent estimator of  $C_{\text{opt}}$.
Then, the GENMETA estimator with $\hat{C}_{\text{opt}}$ is
\begin{equation*}
\hat{\bd{\beta}}_{\text{opt}}
= {\argmin}_{\bd{\beta}}
U^T_n(\bd{\beta}, \hat{\bd{\theta}})\hat{C}_{\text{opt}}U_n(\bd{\beta}, \hat{\bd{\theta}}).
\end{equation*}
Under regularity conditions similar to those in Theorem 1,
$\hat{\bd{\beta}}_{\text{opt}} \rightarrow \bd{\beta}^*$ in probability.
By Mean Value Theorem,
\begin{equation}\label{eqn:aseq1}
U_n(\hat{\bd{\beta}}_{\text{opt}}, \hat{\bd{\theta}})
=
U_n(\bd{\beta}^*, \hat{\bd{\theta}})
+
G_n(\bar{\bd{\beta}}, \hat{\bd{\theta}})(\hat{\bd{\beta}}_{\text{opt}} - \bd{\beta}^*),
\end{equation}
where $\bar{\bd{\beta}}$ is the mean value
and $G_n(\bar{\bd{\beta}}, \hat{\bd{\theta}})
= {\partial}U_n(\bd{\beta}, \hat{\bd{\theta}})/{\partial \bd{\beta}}\mid_{\bd{\beta} = \bar{\bd{\beta}}}$.
By the first order condition, $\hat{\bd{\beta}}_{\text{opt}}$ satisfies
$G^T_n(\hat{\bd{\beta}}_{\text{opt}},\hat{\bd{\theta}})\hat{C}_{\text{opt}}U_n(\hat{\bd{\beta}}_{\text{opt}}, \hat{\bd{\theta}}) = 0$.
Left-multiplying (\ref{eqn:aseq1}) by $G^T_n(\hat{\bd{\beta}}_{\text{opt}}, \hat{\bd{\theta}})\hat{C}_{\text{opt}}$,
it follows
\begin{equation}\label{eqn:aseq2}
\hat{\bd{\beta}}_{\text{opt}} - \bd{\beta}^*
=
-
\{G^T_n(\hat{\bd{\beta}}_{\text{opt}},\hat{\bd{\theta}})\hat{C}_{\text{opt}}\vec{G}_n(\bar{\bd{\beta}}, \hat{\bd{\theta}})\}^{-1} \{G^T_n(\hat{\bd{\beta}}_{\text{opt}},\hat{\bd{\theta}})\hat{C}_{\text{opt}}U_n(\bd{\beta}^*, \hat{\bd{\theta}})\}
\end{equation}
Also,
\begin{equation*}
G_n(\hat{\bd{\beta}}_{\text{opt}}, \hat{\bd{\theta}})
= \frac{\partial }{\partial \bd{\beta}}\vec{U}_n(\bd{\beta}, \hat{\bd{\theta}})\mid_{\bd{\beta} = \hat{\bd{\beta}}_{\text{opt}}}
=
\begin{pmatrix}
\frac{\partial }{\partial \bd{\beta}}u_1(\bd{\beta}, \hat{\bd{\theta}}_1)\mid_{\bd{\beta} = \hat{\bd{\beta}}_{\text{opt}}}\\
\vdots \\
\frac{\partial }{\partial \bd{\beta}}u_K(\bd{\beta}, \hat{\bd{\theta}}_K)\mid_{\bd{\beta} = \hat{\bd{\beta}}_{\text{opt}}}
\end{pmatrix}.
\end{equation*}
Under regularity conditions similar to those in Theorem 1,
${\partial }u_k(\bd{\beta}, \hat{\bd{\theta}}_k)/{\partial \bd{\beta}}\mid_{\bd{\beta} = \hat{\bd{\beta}}_{\text{opt}}}
=\bd{\Sigma}^{-1}+o_{p}(1)$ for each $k$.
Then,
\begin{equation}\label{eqn:Gn:beta.hat}
G_n(\hat{\bd{\beta}}_{\text{opt}}, \hat{\bd{\theta}})
= \begin{pmatrix}
\bd{\Sigma}^{-1}\\
\vdots \\
\bd{\Sigma}^{-1}
\end{pmatrix} + o_p(1).
\end{equation}
Similarly,
\begin{equation}\label{eqn:Gn:beta.bar}
G_n(\bar{\bd{\beta}}, \hat{\bd{\theta}})
= \begin{pmatrix}
\bd{\Sigma}^{-1}\\
\vdots \\
\bd{\Sigma}^{-1}
\end{pmatrix} + o_p(1).
\end{equation}
On the other hand,
under regularity conditions similar to those in Theorem 1,
$u_k(\bd{\beta}^{*}, \hat{\bd{\theta}}_k)
=
-
\bd{\Sigma}^{-1}(\hat{\bd{\theta}}_{k}-\bd{\beta)}^{*} + o_{p}(1/n^{1/2})$.
Then,
\begin{equation}\label{eqn:Un:beta.star}
U_n(\bd{\beta}^{*}, \hat{\bd{\theta}})
=
-
\begin{pmatrix}
\bd{\Sigma}^{-1}(\hat{\bd{\theta}}_{1}-\bd{\beta}^{*})\\
\vdots \\
\bd{\Sigma}^{-1}(\hat{\bd{\theta}}_{K}-\bd{\beta}^{*})
\end{pmatrix} + o_p(1/n^{1/2}).
\end{equation}
Hence,
by (\ref{eqn:aseq2}), (\ref{eqn:Gn:beta.hat}), (\ref{eqn:Gn:beta.bar}), (\ref{eqn:Un:beta.star})
and Slutsky's theorem,
\begin{equation}\label{eqn:beta.hat.opt.minus.beta.star}
\hat{\bd{\beta}}_{\text{opt}} - \bd{\beta}^*
= \Big ( \sum_{k=1}^K{c_k}\Big)^{-1}\Big\{\sum_{k = 1}^K{c_k}(\hat{\bd{\theta}}_k- \bd{\beta}^*)\Big\}  + o_p(1/n^{1/2}).\\
\end{equation}
On the other hand,
\begin{align}\label{eqn:beta.meta.minus.beta.star}
\hat{\bd{\beta}}_{\text{meta}} - \bd{\beta}^*
&=
\Big\{\sum_{k=1}^K\Big(\frac{\hat{\bd{\Sigma}}_k}{n_k}\Big)^{-1}\Big \}^{-1}
\Big \{ \sum_{k = 1}^K\Big(\frac{\hat{\bd{\Sigma}}_k}{n_k}\Big)^{-1}\hat{\bd{\theta}}_k\Big \} - \bd{\beta}^* \nonumber \\
&=
\Big(\sum_{k=1}^K c_{k}\Big)^{-1}
\Big \{ \sum_{k = 1}^K c_{k}(\hat{\bd{\theta}}_k- \bd{\beta}^*) \Big \} + o_{p}(1/n^{1/2}).
\end{align}
Therefore, by (\ref{eqn:beta.hat.opt.minus.beta.star}) and (\ref{eqn:beta.meta.minus.beta.star}),
$\hat{\bd{\beta}}_{\text{opt}} = \hat{\bd{\beta}}_{\text{meta}} + o_p(1/n^{1/2})$.

\subsection{Newton-Raphson's Method and Iteratively Reweighted Least Squares Algorithm}
In this section we provide a derivation of the Newton-Raphson's method for GENMETA with generalized linear models.
As in Section 2.3,
we assume that the maximal and reduced models belong to the class of GLM \citep{mccullagh}.
Specifically,
assume the densities of $Y\mid\bd{X}$ and $Y\mid\bd{X}_{A_{k}}$ are of the forms
$$f(y\mid x; \bd{\beta}, \phi)=
\exp(\{1/a(\phi)\}(yh(\bd{x}^{T}\bd{\beta}) - b\{h(\bd{x}^{T}\bd{\beta})\}) + c(y;\phi)),$$
and
$$g_{k}(y\mid x_{A_k};\bd{\theta}_k)=
\exp(\{1/a(\phi_{k})\}(yh(\bd{x}_{A_{k}}^{T}\bd{\theta}_{k}) - b\{h(\bd{x}_{A_{k}}^{T}\bd{\theta}_{k})\}) + c(y;\phi_{k})),$$
respectively,
where $a(\cdot)$, $b(\cdot)$ and $c(\cdot)$ are known functions,
$h(\cdot)=b'^{-1}(g^{-1}(\cdot))$,
$g$ is a monotone and differentiable link function,
and $\phi$ and $\phi_{k}$ are the dispersion parameters
of the maximal and the $k$th reduced models, respectively.
Recall that we assume the maximal and the reduced models have the same link function $g$.
However, both the GENMETA and the Newton-Raphson's method are flexible to allow the maximal and the reduced models to have different link functions.
We also assume $\vec{X} = \cup_{k=1}^K\vec{X}_{A_k}$,
where the vectors of the covariates are viewed as sets without confusion.
Denote the dimensions of $\bd{\theta}_{k}$ and $\bd{\beta}$
as $d_k$ and $p$, respectively.
Assume $d=\sum_{k=1}^{K}d_k\geq p$
since the parameters of the maximal model will not be identifiable if $d<p$.

\subsubsection{Case I : $\phi$ and $\phi_{k}$'s are known.}

The log-likelihood of $g_{k}$ is
\begin{equation*}
l_{k}(y\mid\bd{x}_{A_{k}}; \bd{\theta}_{k})
=
\{1/a(\phi_{k})\}(yh(\bd{x}_{A_{k}}^{T}\bd{\theta}_{k}) - b\{h(\bd{x}_{A_{k}}^{T}\bd{\theta}_{k})\})+ c(y;\phi_{k}).
\end{equation*}
Then, the score function is
\begin{equation*}
\bd{s}_{k}(y\mid\bd{x}_{A_{k}}; \bd{\theta}_{k})
= \{1/a(\phi_{k})\}
\{y-g^{-1}(\bd{x}_{A_k}^T\bd{\theta}_k)\}h'(\bd{x}_{A_k}^T\bd{\theta}_k)\bd{x}_{A_k}.
\end{equation*}
Then,
\begin{equation*}
\bd{u}_{k}(\bd{x}; \bd{\beta}, \bd{\theta}_{k})
= E_{Y\mid X}\bd{s}_{k}\{(y\mid\bd{x}_{A_{k}}; \bd{\theta}_{k})\}
= \{1/a(\phi_{k})\}
\{g^{-1}(\bd{x}^{T}\bd{\beta})-g^{-1}(\bd{x}_{A_k}^T\bd{\theta}_k)\}h'(\bd{x}_{A_k}^T\bd{\theta}_k)\bd{x}_{A_k}.
\end{equation*}
Thus, the vector of empirical moment functions for $\bd{\beta}$ is
\begin{equation*}
U_n(\bd{\beta}) =
P_n
\begin{pmatrix}
\bd{u}_{k}(\bd{X}; \bd{\beta}, \hat{\bd{\theta}}_{k}) \\
\bd{u}_{k}(\bd{X}; \bd{\beta}, \hat{\bd{\theta}}_{k}) \\
\vdots \\
\bd{u}_{k}(\bd{X}; \bd{\beta}, \hat{\bd{\theta}}_{k})
\end{pmatrix},
\end{equation*}
where $P_n$ is the empirical measure with respect to the reference sample.

Let $Q_n(\bd{\beta}) = U^T_n(\bd{\beta})CU_n(\bd{\beta})$
where $\vec{C}$ is a $d \times d$ positive definite matrix.
The goal is to find the minimizer of $Q_n(\bd{\beta})$.
Its equivalent to solving the equation
\begin{equation*}
D_n(\bd{\beta}) = 0,
\end{equation*}
where $D_n(\bd{\beta})
= \bd{G}_{n}^{T}(\bd{\beta})\bd{C}U_n(\bd{\beta})$ and
$ \bd{G}_{n}(\bd{\beta})
= {\partial}U_n(\bd{\beta})/{\partial \bd{\beta}}$
is a $d \times p$ matrix.
Then, the $t$th iteration step for the Newton-Raphson's method is
\begin{equation}\label{eqn:NR}
\bd{\beta}^{(t+1)} = \bd{\beta}^{(t)} - J_n(\bd{\beta}^{(t)})^{-1} D_n(\bd{\beta}^{(t)}),
\end{equation}
where $J_n(\bd{\beta}) ={\partial}D_n(\bd{\beta})/{\partial \bd{\beta}}$
is a $p \times p$ matrix.\\

Next, we write $D_n(\bd{\beta})$ in a matrix form.
The matrix form of $G_n(\bd{\beta})$ is
\begin{equation*}
\begin{split}
G_n(\bd{\beta})
& =
P_n
\begin{pmatrix}
[a(\phi_{1})g'\{g^{-1}(\bd{X}^{T}\bd{\beta})\}]^{-1}h'(\bd{X}_{A_1}^{T}\hat{\bd{\theta}}_1)\bd{X}_{A_1}\bd{X}^{T} \\
\vdots \\
[a(\phi_{K})g'\{g^{-1}(\bd{X}^{T}\bd{\beta})\}]^{-1}h'(\bd{X}_{A_K}^{T}\hat{\bd{\theta}}_K)\bd{X}_{A_K}\bd{X}^{T}
\end{pmatrix}
= (1/n)X^T_{A_{diag}}WX_{rbind},
\end{split}
\end{equation*}
where $X_{\text{rbind}} = \bd{1} \otimes X$ and $X_{(n\times p)}$ is the reference data matrix;
$X_{A_{\text{diag}}} = \text{diag}(X_{A_1}, \hdots, X_{A_K})$
and $X_{A_k (n\times d_{k})}$ is the reference data matrix for the $k$th study;
$W = \diag(\bd{W}_{1}, \ldots, \bd{W}_{K})$,
$\bd{W}_{k} = \text{diag}(w_{k1},\dots,w_{kn})$,
$w_{ki}
=
[a(\phi_{k})g'\{g^{-1}(\bd{X}_{i}^{T}\bd{\beta})\}]^{-1}h'(\bd{X}_{A_k,i}^{T}\hat{\bd{\theta}}_k)$ for $k = 1,\dots,K$, $i = 1,\dots,n$ and $i$,
and
$\bd{X}_{i}^{T}$ and $\bd{X}_{A_k,i}^{T}$
are the $i$th rows of $X$ and $X_{A_k}$, respectively.
Similarly, the matrix form of $U_n(\bd{\beta})$ is
$U_n(\bd{\beta}) =
(1/n)X^T_{A_{diag}}r$,
where
$r = (r_1, \ldots, r_K)^T$,
$r_k = (r_{k1}, \ldots, r_{kn})^T$
and
$r_{ki} =
\{1/a(\phi_{k})\}\{g^{-1}(X_i^T\bd{\beta}) - g^{-1}(X^T_{A_k,i}\hat{\bd{\theta}}_{A_k,i})\}
h'(X^T_{A_k,i}\hat{\bd{\theta}}_{A_k,i})$ for each $k$ and $i$.
Thus, the matrix form of $D_n(\bd{\beta})$ is
\begin{equation}\label{eqn:Dn}
\begin{split}
D_n(\bd{\beta}) & = (1/n^{2})X^T_{rbind}WX_{A_{diag}}CX^T_{A_{diag}}\bd{r}.
\end{split}
\end{equation}

Next, we write $J_n(\bd{\beta})$ in a matrix form.
Let $G_n(\bd{\beta})$ be partitioned by columns
as $G_n(\bd{\beta}) = (G_{n,1}(\bd{\beta}), \dots, G_{n,p}(\bd{\beta}))$,
where $G_{n,j}(\bd{\beta})$ is a $d \times 1$ column vector for $j = 1, \ldots, p$.
Then,
\begin{equation}\label{eqn:Jn1}
\begin{split}
J_n(\bd{\beta})
& = \frac{\partial }{\partial \bd{\beta}} D_n(\bd{\beta})
= \frac{\partial }{\partial \bd{\beta}} G_n^T(\bd{\beta})CU_n(\bd{\beta}) \\
& = \begin{pmatrix}
\frac{\partial }{\partial \bd{\beta}}G_{n,1}^T(\bd{\beta})CU_n(\bd{\beta})\\
\vdots \\
\frac{\partial }{\partial \bd{\beta}}G_{n,p}^T(\bd{\beta})CU_n(\bd{\beta})\\
\end{pmatrix}
=
G_{n}^{T}(\bd{\beta})CG_n(\bd{\beta})
+
\begin{pmatrix}
U_n^T(\bd{\beta})C\frac{\partial  }{\partial \bd\beta}G_{n,1}(\bd{\beta})\\
\vdots \\
U_n^T(\bd{\beta})C\frac{\partial  }{\partial \bd\beta}G_{n,p}(\bd{\beta})
\end{pmatrix}.
\end{split}
\end{equation}
Then, the matrix form of the first summand
is $(1/n^{2})X_{rbind}^TWX_{A_{diag}}CX^T_{A_{diag}}WX_{rbind}$.
The $j$th row of the second summand is $r^TX_{A_{diag}}C{\partial }G_{n,j}(\bd{\beta})/{\partial \bd{\beta}}$.
Note that
\begin{equation*}
\frac{\partial }{\partial \bd \beta}G_{n,j}(\bd{\beta})
= (1/n)X^T_{A_{diag}} L
   X^*_{j_{diag}}X_{rbind},
\end{equation*}
where
$L = \diag(L_1,\dots,L_K)$,
$L_k = \diag(l_{k1},\dots,l_{kn})$
and, for each $k$ and $i$,
$$l_{ki}
=
-g''\{g^{-1}(\bd{X}_{i}^{T}\bd{\beta})\}/(a(\phi_{k})[g'\{g^{-1}(\bd{X}_{i}^{T}\bd{\beta})\}]^{3}
h'(X^T_{A_k,i}\hat{\bd{\theta}}_{k}));$$
$X^*_{j_{diag}} = \diag(X_{j_{diag}}, \ldots, X_{j_{diag}})$
with $K$ diagonal blocks and
$X_{j_{diag}} = \diag(X_{1j},\ldots,X_{nj})$ for $ j = 1, \ldots, p$.
Then, for each $j$,
the matrix form of
$U_n^T(\bd{\beta})C{\partial}G_{n,j}(\bd{\beta})/{\partial \bd\beta}$ is
$$(1/n^2)r^TX_{A_{diag}}CX^T_{A_{diag}}LX^*_{j_{diag}}X_{rbind}.$$
Then, the second summand of (\ref{eqn:Jn1}) can be rewritten as
$(1/n^2)X_{rbind}^TVX_{rbind}$,
where $V = diag(v_1,\dots,v_{nK})$ and $v_i$ is the $i$th element of the row vector $r^TX_{A_{diag}}CX^T_{A_{diag}}L$.
Thus,
\begin{equation}\label{eqn:Jn2}
J_n(\bd{\beta})
= (1/n^2)X_{rbind}^T(WX_{A_{diag}}CX^T_{A_{diag}}W \ + \ V)X_{rbind}
= (1/n^2)X_{\text{rbind}}^TW^*X_{\text{rbind}}.
\end{equation}
where $W^* = WX_{A_{diag}}C X^T_{A_{diag}}W+V $.\\

Therefore, plugging (\ref{eqn:Dn}) and (\ref{eqn:Jn2}) in (\ref{eqn:NR}),
we get the following $t$th iteration step
\begin{equation*}
\bd \beta^{(t+1)}  = \bd \beta^{(t)} - (X_{\text{rbind}}^TW^*X_{\text{rbind}})^{-1}X^T_{\text{rbind}}
WX_{A_{\text{diag}}}CX^T_{A_{\text{diag}}}r,
\end{equation*}
which can be seen as the $t$th step of an iteratively reweighted least squares algorithm.

\subsubsection{Case II : $\phi$ and $\phi_{k}$'s are unknown.}
When $\phi$ and $\phi_{k}$'s are unknown,
we propose to first obtain the GENMETA estimator $\hat{\bd{\beta}}$ of $\bd{\beta}^{\star}$ as above
with $\phi_{k}'s$ replaced by $\hat{\phi}_{k}$'s.
Next, let us consider the estimation of $\phi^{\star}$, the true value of $\phi$.
For the $k$th reduced model,
we have an additional score function with respect to $\phi_{k}$,
which is
\begin{equation*}
s_{k}(y\mid\bd{x}_{A_{k}}; \bd{\theta}_{k}, \phi_k)
= -\frac{a^\prime(\phi_k)}{a^2(\phi_k)}
(yh(\bd{x}_{A_{k}}^{T}\bd{\theta}_{k}) - b\{h(\bd{x}_{A_{k}}^{T}\bd{\theta}_{k})\}) + c^\prime(y;\phi_k),
\end{equation*}
where $c'(y;\phi_k)$ is the derivative of $c(y;\phi_k)$ with respect to $\phi_{k}$.
Then,
we obtain
\begin{equation*}
u_{k}(X;\bd{\beta}, \phi, \bd{\theta}_{k}, \phi_{k})=
-\frac{a^\prime(\phi_k)}{a^2(\phi_k)}
(g^{-1}(X^T\bd{\beta})h(\bd{X}_{A_{k}}^{T}\bd{\theta}_{k}) - b\{h(\bd{X}_{A_{k}}^{T}\bd{\theta}_{k})\})
+ q_{k}(X; \bd{\beta}, \phi, \phi_k),
\end{equation*}
where $q_{k} = E_{Y\mid\bd{X}}(c'(Y,\phi_k))$.
The distribution of $Y\mid X$ depends on $\bd{\beta}$ and $\phi$
so that $q_{k}$ also depends on them.
Then,
the empirical moment vector for $\phi$ is
$$U_n(\phi) =
P_n
(
u_{1}(X; \hat{\bd{\beta}}, \phi, \hat{\bd{\theta}}_{1}, \hat{\phi}_{1})^{T},
\ldots,
u_{K}(X; \hat{\bd{\beta}}, \phi, \hat{\bd{\theta}}_{K}, \hat{\phi}_{K})^{T}
)^{T}.$$
We propose to estimate $\phi^{\star}$ in the GMM framework.
Thus, we need to compute the minimizer of $U_n(\phi)^{T}\bd{C} U_n(\phi)$,
where $C$ is a known weighting matrix.
As before, we use the Newton-Raphson's method and
it can be written as
\begin{equation} \label{eqn:IRLS iter dip}
\phi^{(t+1)}
= \phi^{(t)} - J_n^{-1}(\phi^{(t)})
D_{n}(\phi^{(t)}),
\end{equation}
where
$$J_n(\phi) = U^T_n(\phi)C\frac{d^2 }{d\phi^2} q_n(\phi)
+ (\frac{d}{d\phi}q_n(\phi))^TC\frac{d}{d\phi}q_n(\phi),$$
$D_{n}(\phi)=U^T_n(\phi^{(t)})C{d}q_n(\phi)/{d\phi}$
and $\vec{q}_n(\phi) =
P_n
(
q_{1}(X; \hat{\bd{\beta}}, \phi, \hat{\phi}_1), \ldots, q_{K}(X; \hat{\bd{\beta}}, \phi, \hat{\phi}_K)
)^T$. \\

Thus, when $\phi$ and $\phi_{k}$'s are unknown,
we first choose initial estimates $\bd \beta^{(0)}$ and $\phi^{(0)}$.
Then, we get the GENMETA estimator $\hat{\bd \beta}$ by using equation (\ref{eqn:NR}) until a stopping rule is reached.
Subsequently, $\phi^{(0)}$, $\hat{\bd \beta}$ and the study estimates are plugged in equation (\ref{eqn:IRLS iter dip})
and the process is repeated until a stopping rule is reached to get the GENMETA estimator of $\phi^{*}$.
In each Newton-Raphson's step, the weighting matrix $\bd{C}$ is estimated by the estimates from the previous step.\\

If the estimates of the study dispersion parameters, $\phi_k$'s, are not provided directly,
but the the outcomes are standardized ($\text{var}(Y) = 1$),
we can obtain them through the following relation based on conditional variance formula
\begin{equation*}
a(\hat{\phi}_k)
=
\frac{1 - (P_n g^{-1}(X^T_{A_k}\hat{\bd{\theta}}_k)^{2} - \{P_n g^{-1}(X^T_{A_k}\hat{\bd{\theta}}_k)\}^{2})}
{P_n b''\{h(X^T_{A_k}\hat{\bd{\theta}}_k)\}},
\end{equation*}
where $h(\cdot)=b'^{-1}(g^{-1}(\cdot))$
and $P_n$ is the empirical measure with the reference data.
For normal family where the canonical link is an identity function,
we have $b^{\prime \prime}(\psi) = 1$,
which implies the denominator is 1.

\subsection{Full Proof of Theorem 1 and
Checking Regularity Assumptions in Two Examples}

We first provide a complete proof of Theorem 1
and then check the assumptions for logistic and linear regression models.
\begin{proof}[Proof of Theorem 1]
  First, we show the consistency of $\hat{\bd{\beta}}$.
Denote
$\hat{\bd{\theta}}$ and $\bd{\theta}^{*}$
as stacked vectors of
$\hat{\bd{\theta}}_{k}$'s and $\bd{\theta}_{k}^{*}$'s, respectively.
Denote
$\bd{U}_{0}(\bd{\beta},\bd{\theta})=E (\bd{U}(\bd{X};\bd{\beta},\bd{\theta}))$
and $\bd{Q}_{0}(\bd{\beta})=\bd{U}_{0}(\bd{\beta},\bd{\theta}^{*})^{T}\bd{C} \bd{U}_{0}(\bd{\beta},\bd{\theta}^{*})$.

By (A1) and Lemma 2.3 of \citep{Newey1994},
$\bd{Q}_{0}(\bd{\beta})$ is uniquely minimized at $\bd{\beta}^{*}$.

By (A2), (A3), (A4) and Lemma 2.4 of \citep{Newey1994},
$\bd{U}_{0}(\bd{\beta},\bd{\theta})$ is continuous
and $\bd{U}_{n}(\bd{\beta}, \bd{\theta})$ converges uniformly to $\bd{U}_{0}(\bd{\beta},\bd{\theta})$
for $(\bd{\beta},\bd{\theta}) \in D_{\bd{\beta}} \times N_{c}(\bd{\theta}^{*})$,
where $N_{c}(\bd{\theta}^{*})$ is a compact subset of $N(\bd{\theta}^{*})$
including $\bd{\theta}^{*}$.
Note that $\hat{\bd{\theta}}$ is a consistent estimator of $\bd{\theta}^{*}$.
With probability going to one (wpg1),
$$\sup_{\bd{\beta}\in D_{\bd{\beta}}}||\bd{U}_{n}(\bd{\beta},\hat{\bd{\theta}})-\bd{U}_{0}(\bd{\beta},\hat{\bd{\theta}})||
\leq \sup_{(\bd{\beta},\bd{\theta}) \in D_{\bd{\beta}} \times N_{c}(\bd{\theta}^{*})}||\bd{U}_{n}(\bd{\beta},\bd{\theta})-\bd{U}_{0}(\bd{\beta},\bd{\theta})||.$$
Then, $\bd{U}_{n}(\bd{\beta},\hat{\bd{\theta}})-\bd{U}_{0}(\bd{\beta},\hat{\bd{\theta}})$
converges uniformly in probability
to 0 for $\bd{\beta}\in D_{\bd{\beta}}$.

For any $r>0$, wpg1,
$$
\sup_{\bd{\beta} \in D_{\bd{\beta}}}||\bd{U}_{0}(\bd{\beta},\hat{\bd{\theta}})-\bd{U}_{0}(\bd{\beta},\bd{\theta}^{*})||
\leq
\sup_{\bd{\beta} \in D_{\bd{\beta}}}E   (\sup_{||\bd{\theta}-\bd{\theta}^{*}||<r}||\bd{U}(\bd{\beta},\bd{\theta})-\bd{U}(\bd{\beta},\bd{\theta}^{*})||).
$$
By (A3), (A4) and dominant convergence theorem,
$E (\sup_{||\bd{\theta}-\bd{\theta}^{*}||<r}||\bd{U}(\bd{\beta},\bd{\theta})-\bd{U}(\bd{\beta},\bd{\theta}^{*})||)$
converges to 0 for every $\bd{\beta}\in D_{\bd{\beta}}$ as $r$ decreases to 0.
Note that
$E (\sup_{||\bd{\theta}-\bd{\theta}^{*}||<r}||\bd{U}(\bd{\beta},\bd{\theta})-\bd{U}(\bd{\beta},\bd{\theta}^{*})||)$
decreases as $r$ decreases for each $\bd{\beta}$.
By (A2) and Dini's theorem (see, for example, Theorem 7.13 of \citep{Rudin1976}),
$E (\sup_{||\bd{\theta}-\bd{\theta}^{*}||<r}||\bd{U}(\bd{\beta},\bd{\theta})-\bd{U}(\bd{\beta},\bd{\theta}^{*})||)$
converges uniformly in probability
to 0 for $\bd{\beta}\in D_{\bd{\beta}}$ as $r$ decreases to 0.
Then,
$\bd{U}_{0}(\bd{\beta},\hat{\bd{\theta}})-\bd{U}_{0}(\bd{\beta},\bd{\theta}^{*})$
converges uniformly in probability
to 0 for $\bd{\beta}\in D_{\bd{\beta}}$.

By combining the above two results,
it follows that
$\bd{U}_{n}(\bd{\beta},\hat{\bd{\theta}})$ converges uniformly in probability
to $\bd{U}_{0}(\bd{\beta},\bd{\theta}^{*})$ for $\bd{\beta}\in D_{\bd{\beta}}$.

By the triangle and Cauchy-Schwartz inequalities,
\begin{align*}
    \sup_{\bd{\beta}\in D_{\bd{\beta}}}|\bd{Q}_{n}(\bd{\beta}) - \bd{Q}_{0}(\bd{\beta})|
    & \leq ||\hat{\bd{C}}||\sup_{\bd{\beta}\in D_{\bd{\beta}}}||\bd{U}_{n}(\bd{\beta}, \hat{\bd{\theta}})-\bd{U}_{0}(\bd{\beta},\bd{\theta}^{*})||^{2} \\
    & + 2||\hat{\bd{C}}||\sup_{\bd{\beta}\in D_{\bd{\beta}}}||\bd{U}_{0}(\bd{\beta},\bd{\theta}^{*})||
\sup_{\bd{\beta}\in D_{\bd{\beta}}}||\bd{U}_{n}(\bd{\beta}, \hat{\bd{\theta}})-\bd{U}_{0}(\bd{\beta},\bd{\theta}^{*})|| \\
    & + ||\hat{\bd{C}}-\bd{C}||\sup_{\bd{\beta}\in D_{\bd{\beta}}}||\bd{U}_{0}(\bd{\beta},\bd{\theta}^{*})||^{2}
\end{align*}
Since $\hat{\bd{C}}$ is a consistent estimator of $\bd{C}$,
$||\hat{\bd{C}}||$ converges in probability to $||\bd{C}||$,
which is finite; $||\hat{\bd{C}}-\bd{C}||$ converges in probability to 0.
Since $\bd{U}_{0}(\bd{\beta},\bd{\theta}^{*})$ is continuous for $\bd{\beta}\in D_{\bd{\beta}}$
and $D_{\bd{\beta}}$ is compact,
$\sup_{\bd{\beta}\in D_{\bd{\beta}}}||\bd{U}_{0}(\bd{\beta},\bd{\theta}^{*})||^{2}$ is finite.
Since $\sup_{\bd{\beta}\in D_{\bd{\beta}}}||\bd{U}_{n}(\bd{\beta}, \hat{\bd{\theta}})-\bd{U}_{0}(\bd{\beta},\bd{\theta}^{*})||$ converges in probability to 0,
$\sup_{\bd{\beta}\in D_{\bd{\beta}}}||\bd{U}_{n}(\bd{\beta}, \hat{\bd{\theta}})-\bd{U}_{0}(\bd{\beta},\bd{\theta}^{*})||^{2}$ converges in probability to 0.
Thus, $\bd{Q}_{n}(\bd{\beta})-\bd{Q}_{0}(\bd{\beta})$ converges uniformly in probability to 0 for $\bd{\beta}\in D_{\bd{\beta}}$.
Recall that $\bd{\beta}^{*}$ is the unique minimizer of $\bd{Q}_{0}(\bd{\beta})$.
By Theorem 2.1 of \citep{Newey1994},
$\hat{\bd{\beta}}$ is a consistent estimator of $\bd{\beta}^{*}$.\\

Next, we derive the asymptotic distribution of the GENMETA estimator $\hat{\bd{\beta}}$.
Note that $\hat{\bd{\beta}}$ is a solution to
$$\bd{G}_{n}(\bd{\beta}, \hat{\bd{\theta}})^{T}\hat{\bd{C}}\bd{U}_{n}(\bd{\beta}, \hat{\bd{\theta}})=0,$$
where
$\bd{G}_{n}(\bd{\beta}, \hat{\bd{\theta}})={\partial}\bd{U}_{n}(\bd{\beta}, \hat{\bd{\theta}})/{\partial \bd{\beta}}$,
the Jacobian of $\bd{U}_{n}(\bd{\beta}, \hat{\bd{\theta}})$.
On the other hand,
by mean value theorem,
$$\bd{U}_{n}(\hat{\bd{\beta}}, \hat{\bd{\theta}})=\bd{U}_{n}(\bd{\beta}^{*}, \hat{\bd{\theta}})
+ \bd{G}_{n}(\bar{\bd{\beta}}, \hat{\bd{\theta}})(\hat{\bd{\beta}}-\bd{\beta}^{*}),$$
where $\bar{\bd{\beta}}$ denotes a matrix
each column of which corresponds to each element of $\bd{U}_{n}(\bd{\beta}, \hat{\bd{\theta}})$.
After left multiplying $\bd{G}_{n}(\hat{\bd{\beta}}, \hat{\bd{\theta}})^{T}\hat{\bd{C}}$ to the above identity,
it follows
$$n^{1/2}(\hat{\bd{\beta}}-\bd{\beta}^{*})
=
-\bd{M}_{n}n^{1/2}\bd{U}_{n}(\bd{\beta}^{*}, \hat{\bd{\theta}}),$$
where
$\bd{M}_{n}=(\bd{G}_{n}(\hat{\bd{\beta}}, \hat{\bd{\theta}})^{T}\hat{\bd{C}}\bd{G}_{n}(\bar{\bd{\beta}}, \hat{\bd{\theta}}))^{-1}
\bd{G}_{n}(\hat{\bd{\beta}}, \hat{\bd{\theta}})^{T}\hat{\bd{C}}$.

Consider $\bd{M}_{n}$.
Since
$\hat{\bd{\beta}}$ is a consistent estimator of $\bd{\beta}^{*}$,
each column of $\bar{\bd{\beta}}$ is a consistent estimator of $\bd{\beta}^{*}$.
On the other hand, $\hat{\bd{\theta}}$ is a consistent estimator of $\bd{\theta}^{*}$.
By (A5), (A6) and Lemma 2.4 of \citep{Newey1994},
$\bd{G}_{n}(\bd{\beta}, \bd{\theta})$ converge uniformly to continuous
$E\{{\partial}\bd{U}(\bd{X}; \bd{\beta}, \bd{\theta})/{\partial \bd{\beta}}\}$
for $(\bd{\beta},\bd{\theta})\in D_{\bd{\beta}}\times N_{c}(\bd{\theta}^{*})$,
where $N_{c}(\bd{\theta}^{*})$ is a compact subset of $N(\bd{\theta}^{*})$,
including $\bd{\theta}^{*}$.
Since $\hat{\bd{\beta}}$ and each column of $\bar{\bd{\beta}}$ converge in probability to $\bd{\beta}^{*}$
and $\hat{\bd{\theta}}$ is a consistent estimator of $\bd{\theta}^{*}$,
by, for example, Theorem 9.4 of \citep{Keener2010},
both $\bd{G}_{n}(\hat{\bd{\beta}}, \hat{\bd{\theta}})$
and $\bd{G}_{n}(\bar{\bd{\beta}}, \hat{\bd{\theta}})$ converges in probability to
$\bd{\Gamma}=E\{{\partial}\bd{U}(\bd{X}; \bd{\beta}^{*}, \bd{\theta}^{*})/{\partial \bd{\beta}}\}$.
Thus, by noting $\hat{\bd{C}}{\rightarrow} \bd{C}$ in probability,
$\bd{M}_{n}$ converges in probability to
$(\bd{\Gamma}^{T}\bd{C}\bd{\Gamma})^{-1}\bd{\Gamma}^{T}\bd{C}$.

Consider $n^{1/2}\bd{U}_{n}(\bd{\beta}^{*}, \hat{\bd{\theta}})$.
By mean value theorem,
$$\bd{U}_{n}(\bd{\beta}^{*}, \hat{\bd{\theta}}) = \bd{U}_{n}(\bd{\beta}^{*}, \bd{\theta}^{*})
+ \bd{V}_{n}(\bd{\beta}^{*}, \bar{\bd{\theta}})(\hat{\bd{\theta}}-\bd{\theta}^{*}),$$
where $\bd{V}_{n}$ is the Jacobian of $\bd{U}_{n}(\bd{\beta}^{*}, \bd{\theta})$ as a function of $\bd{\theta}$ and
$\bar{\bd{\theta}}$ is a matrix
each column of which corresponds to each element of $\bd{U}_{n}(\bd{\beta}^{*}, \bd{\theta})$.
Thus,
$$n^{1/2}\bd{U}_{n}(\bd{\beta}^{*}, \hat{\bd{\theta}})
= n^{1/2}\bd{U}_{n}(\bd{\beta}^{*}, \bd{\theta}^{*})
+ \bd{V}_{n}(\bd{\beta}^{*}, \bar{\bd{\theta}})n^{1/2}(\hat{\bd{\theta}}-\bd{\theta}^{*}).$$
By (A9) and central limit theorem,
$n^{1/2}\bd{U}_{n}(\bd{\beta}^{*}, \bd{\theta}^{*})
\overset{d}{\rightarrow}N(0, \bd{\Delta})$.
Since $\hat{\bd{\theta}}$ is a consistent estimator of $\bd{\theta}^{*}$.
each column of $\bar{\bd{\theta}}$ converges in probability to $\bd{\theta}^{*}$.
Similar to the above argument,
by (A7), (A8), Lemma 2.4 of \citep{Newey1994} and Theorem 9.4 of \citep{Keener2010},
$$\bd{V}_{n}(\bd{\beta}^{*}, \bar{\bd{\theta}})
{\rightarrow}
\diag(\bd{W}_{1}, \bd{W}_{2}, \ldots, \bd{W}_{K}) \quad \text{in probability},$$
where, for $k=1,2,\ldots, K$,
$\bd{W}_{k} =  E \{{\partial}
u_{k}(\bd{X}, \bd{\beta}^{*}, \bd{\theta}_{k})/{\partial \bd{\theta}_{k}}\}
\mid_{\bd{\theta}_{k} = \bd{\theta}_{k}^{*}}$.
The $K$ study data sets are independent.
So are $\hat{\bd{\theta}}_{k}$'s.
Note that $n_{k}/n\rightarrow c_{k}$, where $c_{k}$ is a positive constant for $k=1,2,\ldots, K$.
Then
$n^{1/2}(\hat{\bd{\theta}}-\bd{\theta}^{*})$
converges in distribution to
$$N(0, \diag((1/c_{1})\bd{\Sigma}_{1}, (1/c_{2})\bd{\Sigma}_{2},\ldots, (1/c_{K})\bd{\Sigma}_{K})).$$
Since the $K$ data sets
and the reference data are independent, the above results imply that
$n^{1/2}\bd{U}_{n}(\bd{\beta}^{*}, \hat{\bd{\theta}})$
converges in distribution to
$N(0, \bd{\Delta} + \bd{\Lambda})$,
where
$\bd{\Lambda}$ is a block diagonal matrix whose $k$th
block is
$(1/c_{k})\bd{W}_{k}\bd{\Sigma}_{k}\bd{W}_{k}^{T}$
for $k=1,\ldots,K$.

Therefore, with the above two results on $\bd{M}_{n}$ and $n^{1/2}\bd{U}_{n}(\bd{\beta}^{*}, \hat{\bd{\theta}})$ and by Slutsky's theorem,
the asymptotic normality of $n^{1/2}(\hat{\bd{\beta}}-\bd{\beta}^{*})$ follows.
\end{proof}


\begin{example}[Check Assumptions for Logistic Regression Model]
Suppose the maximal model is
$$
Y\mid\bd{X} \sim \text{Bernoulli}\Big\{\frac{1}{1+\exp(-X^{T}\bd{\beta}^{*})}\Big\},
$$
where $X=(1,\bd{X}^{T})^{T}$,
$\bd{X}=({X}_{1}, \ldots, {X}_{d})^{T}$
is the vector of covariates
and
$\bd{\beta}^{*}
=({\beta}_{0}^{*}, {\beta}_{1}^{*}, \ldots,{\beta}_{p}^{*})^{T}$
is the vector of coefficients of interest.
There are $K$ independent studies and
the reduced model of the $k$th study is
$$
Y\mid\bd{X}_{A_{k}}\sim
\text{Bernoulli}\Big\{\frac{1}{1+\exp(-X_{A_{k}}^{T}\bd{\theta}_{k})}\Big\},$$
where $X_{A_{k}}=(1,\bd{X}_{A_{k}}^{T})^{T}$,
$\bd{X}_{A_{k}}$ is a sub-vector of $\bd{X}$
with $A\subset\{1,2,\ldots, p\}$.
For example, $\bd{X}_{A}=({X}_{1}, {X}_{2})^{T}$ when $A=\{1,2\}$.

The global identification assumption (A1) usually holds
and $D_{\bd{\beta}}$ is a compact set.
Next, we check the assumptions (A3) to (A9).
The moment functions from the $k$th study is
$$
u_{k}(\bd{X}; \bd{\beta}, \bd{\theta}_{k})
=
\Big(
\frac{1}{1+e^{-X^{T}\bd{\beta}}}-
\frac{1}{1+e^{-X_{A_{k}}^{T}\bd{\theta}_{k}}}
\Big)
X_{A_{k}}.
$$
It is a continuous function of $\bd{\beta}$ and $\bd{\theta}_{k}$.
Then, (A3) is satisfied.
Note that
$$
\sup_{(\bd{\beta},\bd{\theta}) \in D_{\bd{\beta}} \times N(\bd{\theta}^{*})}
||
\Big (
\frac{1}{1+e^{-X^{T}\bd{\beta}}}-
\frac{1}{1+e^{-X_{A_{k}}^{T}\bd{\theta}_{k}}}
\Big )
X_{A_{k}}
||\leq 2||X||_{1},
$$
where $||\cdot||$ and $||\cdot||_{1}$ are the $l_{2}$ and $l_{1}$ norms, respectively.
Then, given $E(|{X}_{i}|)<\infty$ for each $i$,
(A4) is satisfied.
Also,
\begin{equation}\label{example:logistic:partial:beta}
\frac{\partial}{\partial \bd{\beta}}u_{k}(\bd{X};\bd{\beta}, \bd{\theta}_{k})
= \frac{e^{-X^{T}\bd{\beta}}}{(1+e^{-X^{T}\bd{\beta}})^{2}}X_{A_{k}}X^{T},
\end{equation}
which does not depend on $\bd{\theta}_{k}$ and is continuous for each $\bd{\beta}$.
Then, (A5) is verified.
Note that
$$
\sup_{(\bd{\beta},\bd{\theta}) \in D_{\bd{\beta}} \times N(\bd{\theta}^{*})}
||
\frac{e^{-X^{T}\bd{\beta}}}{(1+e^{-X^{T}\bd{\beta}})^{2}}X_{A_{k}}X^{T}
||
\leq
||XX^{T}||_{1}.
$$
Given $E ({X}_{i}^{2})<\infty$ for each $i$,
(A6) is satisfied.
Note that
$$
\frac{\partial}{\partial \bd{\theta}_{k}}u_{k}(\bd{X}; \bd{\beta}^{*}, \bd{\theta}_{k})
=
-\frac{e^{-X_{A_{k}}^{T}\bd{\theta}_{k}}}{(1+e^{-X_{A_{k}}^{T}\bd{\theta}_{k}})^{2}}X_{A_{k}}X_{A_{k}}^{T},
$$
which is continuous for each $\bd{\theta}_{k}$.
Then, (A7) is satisfied.
Note that
$$
\sup_{(\bd{\beta},\bd{\theta}) \in D_{\bd{\beta}} \times N(\bd{\theta}^{*})}
||
-\frac{e^{-X_{A_{k}}^{T}\bd{\theta}_{k}}}{(1+e^{-X_{A_{k}}^{T}\bd{\theta}_{k}})^{2}}X_{A_{k}}X_{A_{k}}^{T}
||
\leq
||XX^{T}||_{1}.
$$
Given $E ({X}_{i}^{2})<\infty$ for each $i$,
(A8) is satisfied.
The absolute value of each element of $\bd{\Delta}(\bd{\beta}^{*}, \bd{\theta}^{*})$ is less than 1, $E(|{X}_{i}|)$ or $E(|{X}_{i}{X}_{j}|)$
for each $i$ and $j$. Given $E ({X}_{i}^{2})<\infty$, $\bd{\Delta}(\bd{\beta}^{*}, \bd{\theta}^{*})$ is finite.
Note that $\bd{\Gamma}(\bd{\beta}^{*}, \bd{\theta}_{k}^{*})$
is a stacked matrix of (\ref{example:logistic:partial:beta}) for $k=1,\ldots,K$.
Given each covariate of the maximal model is in at least one reduced model
and $E[\{e^{-X^{T}\bd{\beta}}/(1+e^{-X^{T}\bd{\beta}})^{2}\}XX^{T}]$ is positive definite,
$\bd{\Gamma}(\bd{\beta}^{*}, \bd{\theta}^{*})$ is of full rank. Then, (A9) is verified. \qed
\end{example}

\begin{example}[Check Assumptions for Linear Regression Model]
\normalfont
Suppose the true maximal model is
$$Y\mid\bd{X} \sim N(\bd{X}^{T}\bd{\beta}^{*},\sigma^{* 2}),$$
where $\bd{X}=(\bd{X}_{1}, \bd{X}_{2}, \ldots, \bd{X}_{p})^{T}$;
$\bd{\beta}^{*}=(
\bd{\beta}_{1}^{*},
\bd{\beta}_{2}^{*},
\ldots,
\bd{\beta}_{p}^{*})^{T}$;
$E (\bd{X}) = 0$ and $E (Y) = 0$,
that is, both $\bd{X}$ and $Y$ are centered.
There are $K$ independent studies and
the reduced model of the $k$th study is
$$Y\mid\bd{X}_{A_{k}}\sim N(\bd{X}_{A_{k}}^{T}\bd{\theta}_{k}, \sigma_{k}^{2}).$$
For simplicity, assume $\sigma^{* 2}$ is known
and the unknown parameter is $\bd{\beta}^{*}$.
The case with unknown $\sigma^{* 2}$ can be similarly considered.

The moment functions from the $k$th reduced model is
$$u_{k}(\bd{X}; \bd{\beta}; \bd{\theta}_{k}, \sigma_{k}^{2})
= \frac{1}{\sigma_{k}^{2}}(\bd{X}_{A_{k}}\bd{X}^{T}\bd{\beta}  - \bd{X}_{A_{k}}\bd{X}_{A_{k}}^{T}\bd{\theta}_{k}),
$$
which is linear in $\bd{\beta}$.
Note that
\begin{equation}\label{example:linear:partial:beta}
\frac{\partial}{\partial \bd{\beta}}u_{k}(\bd{X}; \bd{\beta}; \bd{\theta}_{k}, \sigma_{k}^{2})
= \frac{1}{\sigma_{k}^{2}}\bd{X}_{A_{k}}\bd{X}^{T}.
\end{equation}
Given each covariate of the maximal model is in at least one reduced model
and $E ({\bd{X}}{\bd{X}}^{T})$ is positive definite,
$\bd{\Gamma}(\bd{\beta}^{*}, \{\bd{\theta}_{k}^{*}\}, \{\sigma_{k}^{* 2}\})
={\partial}u_{k}(\bd{X}; \bd{\beta}^{*}; \{\bd{\theta}_{k}^{*}\}, \{\sigma_{k}^{* 2}\})/{\partial \bd{\beta}}$ is of full rank.
Given $\bd{C}$ is positive definite, (A1) is satisfied.
Suppose $D_{\bd{\beta}}$ is a compact set. Then, (A2) is satisfied.

Next, we check the assumptions (A3) to (A9).
Note that $u_{k}(\bd{X}; \bd{\beta}; \bd{\theta}_{k}, \sigma_{k}^{2})$ is continuous for every $(\bd{\beta}, \bd{\theta}_{k}, \sigma_{k}^{2})$.
Then, (A3) is satisfied.
Note that
$$
\sup_{(\bd{\beta}, \bd{\theta}_{k}, \sigma_{k}^{2})}
||
\frac{1}{\sigma_{k}^{2}}(\bd{X}_{A_{k}}\bd{X}^{T}\bd{\beta}  - \bd{X}_{A_{k}}\bd{X}_{A_{k}}^{T}\bd{\theta}_{k})
||\leq \frac{1}{\sigma_{k}^{2}}(||\bd{\beta}||+||\bd{\theta}_{k}||)||\bd{X}\bd{X}^{T}||_{1},
$$
Denote a finite upper bound of $||\bd{\beta}||$ for $\bd{\beta}\in D_{\bd{\beta}}$ as $\bd{C}(\bd{\beta})$,
a finite upper bound of $||\bd{\theta}_{k}||$ for $\bd{\theta}_{k}\in N(\bd{\theta}_{k}^{*})$ as $\bd{C}(\bd{\theta}_{k})$,
and a positive finite lower bound of $\sigma_{k}^{2}$ for $\sigma_{k}^{2}\in N(\bd{\theta}_{k}^{*})$ as $\sigma_{L}^{2}$.
The supremum of $(1/\sigma_{k}^{2})(||\bd{\beta}||+||\bd{\theta}_{k}||)$
for $(\bd{\beta}, \bd{\theta}_{k}, \sigma_{k}^{2}) \in D_{\bd{\beta}}\times N(\bd{\theta}_{k}^{*})\times N(\sigma_{k}^{* 2})$
is bounded by $(1/\sigma_{L}^{2})(\bd{C}(\bd{\beta})+\bd{C}(\bd{\theta}_{k}))$.
Given $E ({X}_{i}^{2})<\infty$ for each $i$,
(A4) is satisfied.
Note that
$
{\partial}u_{k}(\bd{X}; \bd{\beta}; \bd{\theta}_{k}, \sigma_{k}^{2})/{\partial \bd{\beta}}
$
does not depend on $\bd{\beta}$ and $\bd{\theta}_{k}$
and is continuous for each $\sigma_{k}^{2}$.
Then, (A5) is satisfied.
Note that
$$
\sup_{\sigma_{k}^{2}\in N(\sigma_{k}^{* 2})}
||
\frac{1}{\sigma_{k}^{2}}\bd{X}_{A_{k}}\bd{X}^{T}
||
\leq \frac{1}{\sigma_{L}^{2}}||\bd{X}\bd{X}^{T}||_{1}.
$$
Given $E ({X}_{i}^{2})<\infty$ for each $i$,
(A6) is satisfied.
Note that
$$
\frac{\partial}{\partial (\bd{\theta}_{k}, \sigma_{k}^{2})}u_{k}(\bd{X}; \bd{\beta}; \bd{\theta}_{k}, \sigma_{k}^{2})
=
\{
-\frac{1}{\sigma_{k}^{2}}\bd{X}_{A_{k}}\bd{X}_{A_{k}}^{T},
-\frac{1}{\sigma_{k}^{4}}(\bd{X}_{A_{k}}\bd{X}^{T}\bd{\beta}  - \bd{X}_{A_{k}}\bd{X}_{A_{k}}^{T}\bd{\theta}_{k})
\},
$$
which is continuous for every $(\bd{\beta}, \bd{\theta}_{k}, \sigma_{k}^{2})$.
Then, (A7) is satisfied.
For every $(\bd{\beta}, \bd{\theta}_{k}, \sigma_{k}^{2})\in D_{\bd{\beta}}\times N(\bd{\theta}_{k}^{*}, N(\sigma_{k}^{* 2}))$,
the $l_{2}$ norm of the above partial derivative is less than or equal to
$$\frac{1}{\sigma_{L}^{2}}+\frac{1}{\sigma_{L}^{4}}(\bd{C}(\bd{\beta})+\bd{C}(\bd{\theta}_{k}))||\bd{X}\bd{X}^{T}||_{1}.$$
Given $E ({X}_{i}^{2})<\infty$ for each $i$,
(A8) is satisfied.
Each element of
$\bd{\Delta}(\bd{\beta}^{*}, \{\bd{\theta}_{k}^{*}\}, \{\sigma_{k}^{* 2}\})$
is equal to a constant times $E ({X}_{i_{1}}{X}_{i_{2}}{X}_{i_{3}}{X}_{i_{4}})$
for some $i_{1}, i_{2}, i_{3}, i_{4}$.
Given $E ({X}_{i}^{4})<\infty$ for each $i$,
$\bd{\Delta}$ is finite.
Note that $\bd{\Gamma}(\bd{\beta}^{*}, \{\bd{\theta}_{k}^{*}\}, \{\sigma_{k}^{* 2}\})$
is a stacked matrix of (\ref{example:linear:partial:beta}) for $k=1,\ldots,K$.
As in checking (A2),
given each covariate of the maximal model is in at least one reduced model
and $E ({\bd{X}}{\bd{X}}^{T})$ is positive definite,
$\bd{\Gamma}$ is of full rank. Then, (A9) is verified.\qed
\end{example}

\subsection{Simulation Results for Log-normally Distributed Covariates}
\begin{table}[htb]
\begin{threeparttable}
\caption{Robustness of GENMETA Estimation (Log-normally Distributed Covariates)}{%
\begin{tabular}{ccccccrcccc}
  Setting  & Study-I & Study-II & Study-III & Reference & $\beta_{i}^{*}$ & Bias & SD (ESD) & RMSE & CR & AL \\ 
\vspace{-1mm} &&&&&&&&&& \\
  & $\mu_b$ & $\mu_b$ & $\mu_b$& $\mu_b$ & $\beta_{1}^{*}$ &  .010 & .076 (.075) & .077 & .941 & .288 \\
I  &$\sigma_b^2$ & $\sigma_b^2$ & $\sigma_b^2$ & $\sigma_b^2$ &  $\beta_{2}^{*}$ & .011 & .064 (.061) & .065 & .947 & .237 \\
 & $\rho_b$ & $\rho_b$ & $\rho_b$ & $\rho_b$ & $\beta_{3}^{*}$ & .006 & .066 (.064) & .066 & .954 & .246 \\  
\vspace{-1mm} & & &&&&&&&& \\
  & $\mu_b$ & $\mu_h$ & $\mu_m$ &  $\mu_b$ & $\beta_{1}^{*}$ & .010 & .079 (.072) & .079 & .930 & .272 \\
  & $\sigma_b^2$  & $\sigma_b^2$ & $\sigma_b^2$ & $\sigma_b^2$ & $\beta_{2}^{*}$ &  .002 & .056 (.054) & .056 & .948 & .211 \\
  & $\rho_b$ & $\rho_b$ & $\rho_b$ & $\rho_b$ & $\beta_{3}^{*}$ & -.002 & .062 (.058) & .062 & .945 & .222 \\   
\vspace{-3mm} & &  &&&&&&&& \\
    & $\mu_b$ & $\mu_b$ & $\mu_b$ & $\mu_b$ & $\beta_{1}^{*}$ &  .032 & .088 (.088) & .094 & .930 & .339 \\
 II   & $\sigma_b^2$ & $\sigma_h^2$ & $\sigma_l^2$ & $\sigma_b^2$ &  $\beta_{2}^{*}$ &  -.002 & .062 (.057) & .062 & .941 & .221 \\
  & $\rho_b$ & $\rho_b$ & $\rho_b$ & $\rho_b$ & $\beta_{3}^{*}$ &  -.005 & .074 (.074) & .074 & .967 & .286 \\   
\vspace{-3mm} &&&&&&&&&&\\
   & $\mu_b$ & $\mu_h$ & $\mu_m$ &  $\mu_b$  & $\beta_{1}^{*}$ & .021 & .079 (.077) & .081 & .929 & .294 \\
  & $\sigma_b^2$ & $\sigma_h^2$ & $\sigma_l^2$ & $\sigma_b^2$ & $\beta_{2}^{*}$&  .0005 & .055 (.055) & .055 & .956 & .213 \\
   & $\rho_b$ & $\rho_b$ & $\rho_b$ & $\rho_b$ & $\beta_{3}^{*}$ &  -.008 & .065 (.064) & .065 & .954 & .246 \\   
\vspace{-1mm} & & &&&&&&&& \\
  & $\mu_b$ & $\mu_b$ & $\mu_b$ & $\mu_b$ & $\beta_{1}^{*}$ & -.062 & .107 (.118) & .124 & .934 & .382 \\
 & $\sigma_b^2$ & $\sigma_b^2$ & $\sigma_b^2$ & $\sigma_b^2$  & $\beta_{2}^{*}$&   .021 & .070 (.065) & .073 & .930 & .250 \\
  & $\rho_b$ & $\rho_b$ & $\rho_b$ & $\rho_h$ & $\beta_{3}^{*}$ &  .030 & .087 (.096) & .092 & .956 & .322 \\   
\vspace{-3mm} & & &&&&&&&& \\
  &  $\mu_b$ & $\mu_b$ & $\mu_b$ & $\mu_b$ & $\beta_{1}^{*}$ &  .039 & .072 (.069) & .081 & .891 & .264 \\
 &  $\sigma_b^2$ & $\sigma_b^2$ & $\sigma_b^2$ & $\sigma_b^2$ &  $\beta_{2}^{*}$&  .023 & .065 (.062) & .069 & .932 & .240 \\
  & $\rho_b$ & $\rho_b$ & $\rho_b$ & $\rho_l$ & $\beta_{3}^{*}$ &  .018 & .061 (.058) & .064 & .930 & .224 \\   
  \vspace{-3mm} & & &&&&&&&& \\
   &  $\mu_b$ & $\mu_b$ & $\mu_b$ & $\mu_b$ & $\beta_{1}^{*}$ &  .053 & .079 (.075) & .095 & .866 & .290 \\
 III &  $\sigma_b^2$ & $\sigma_b^2$ & $\sigma_b^2$ & $\sigma_b^2$ &  $\beta_{2}^{*}$&  .019 & .065 (.063) & .067 & .942 & .242 \\
  & $\rho_l$ & $\rho_b$ & $\rho_h$ & $\rho_l$ & $\beta_{3}^{*}$ &  .012 & .068 (.064) & .069 & .935 & .249 \\   
  \vspace{-3mm}  & & &&&&&&&& \\
   &  $\mu_b$ & $\mu_b$ & $\mu_b$ & $\mu_b$ & $\beta_{1}^{*}$ &  .032 & .089 (.084) & .095 & .912 & .322  \\
   &  $\sigma_b^2$ & $\sigma_b^2$ & $\sigma_b^2$ & $\sigma_b^2$ &  $\beta_{2}^{*}$&  .010 & .062 (.062) & .063 & .946 & .240 \\
  & $\rho_l$ & $\rho_b$ & $\rho_h$ & $\rho_b$ & $\beta_{3}^{*}$ &  -.009 & .073 (.071) & .073 & .942 & .273 \\   
  \vspace{-3mm}  & & &&&&&&&& \\
   &  $\mu_b$ & $\mu_b$ & $\mu_b$ & $\mu_b$ & $\beta_{1}^{*}$ &  -.025 & .113 (.108) & .116  & .954 & .407 \\
  &  $\sigma_b^2$ & $\sigma_b^2$ & $\sigma_b^2$ & $\sigma_b^2$ &  $\beta_{2}^{*}$&  .017 & .065 (.064) & .067 & .951 & .248 \\
& $\rho_l$ & $\rho_b$ & $\rho_h$ & $\rho_h$ & $\beta_{3}^{*}$ &  -.002 & .091 (.091) & .091 & .965 & .347 \\   
  \vspace{-1mm}  & & &&&&&&&& \\
 &  & & $\mu_b$ & $\mu_b$ & $\beta_{1}^{*}$ &  .007 & .096 (.104) & .096  & .968 & .365 \\
 IV &  $X_1>-0.5,$ & $X_2 > 0$ & $\sigma_b^2$ & $\sigma_b^2$ &  $\beta_{2}^{*}$&  .242 & .353 (.117) & .428 & .572 & .401 \\
  & $X_2<0.5$ &  & $\rho_b$ & $\rho_b$ & $\beta_{3}^{*}$ &  -.015 & .067 (.081) & .068 & .971 & .283 \\   \hline
\end{tabular}}
\label{table:robustnesslognormal}
\begin{tablenotes}[para,flushleft]
Biases,
standard deviation (SD),
estimated standard deviation (ESD),
square roots of mean square errors (RMSE),
coverage rates (CR),
and average lengths (AL)
of 95\% confidence intervals of the GENMETA estimates using the study covariance estimators in the setting of logistic regression. In setting (I), data are simulated in ideal setting where the covariate distribution is a log-normal distribution with the natural logarithm of the covariates being characterized by mean, sd and correlation of normal variates and are assumed to same across all populations. In setting (II)-(IV), the assumption is violated by creating variations in mean/sd, correlations of the underlying normal distribution and selection criterion across the studies and reference sample. The vector of means, variances and correlations of the underlying normal covariates are denoted by $\mu_*  = (\mu_1,\mu_2,\mu_3)$, $\sigma^2_*  = (\sigma_1^2,\sigma_2^2, \sigma_3^2)$ and $\rho_*  = (\rho_{12},\rho_{23},\rho_{13})$ for $* \in \{b,l,m,h\}$, where $\mu_b = (0,0,0)$, $\mu_m = (0.5,0.5,0.5)$, $\mu_h = (1,1,1)$; $\sigma_b^2 = (1,1,1)$, $\sigma_l^2 = (0.5,0.5,0.5)$, $\sigma_h^2 = (2,2,2)$ and $\rho_b = (0.3,0.6,0.1)$, $\rho_h = (0.4,0.8,0.2)$, $\rho_l = (0.2,0.4,0)$. Estimated standard deviation are obtained by the asymptotic formula (2) in the main paper and used to construct 95\% confidence interval.
\end{tablenotes}
\end{threeparttable}
\end{table}

\end{document}